\documentclass[acmsmall]{acmart}
\renewcommand\footnotetextcopyrightpermission[1]{} 
\usepackage{fancyhdr}
\AtBeginDocument{%
  }

\usepackage{amsthm,amsmath,amsfonts, mathtools}
\usepackage{graphicx}
\usepackage{comment}
\usepackage{caption}
\usepackage{subcaption}
\usepackage{float}
\usepackage{makecell}
\usepackage{dirtytalk}
\usepackage{soul}
\usepackage{bm}
\usepackage{xcolor}
\usepackage{algpseudocode}
\usepackage{algorithm}
\newtheorem{lemma}{Lemma}
\newtheorem{remark}{Remark}

\newcommand{\un}{U}
\newcommand{\load}{L}
\newcommand{\sync}{S}

\newcommand\newtop[2]{\genfrac{}{}{0pt}{}{#1}{#2}}

\newcommand{\specialcell}[2][c]{%
  \begin{tabular}[#1]{@{}#1@{}}#2\end{tabular}}

\algdef{SE}[DOWHILE]{Do}{doWhile}{\algorithmicdo}[1]{\algorithmicwhile\ #1}%

\setcopyright{none}

\begin{document}

\title{The Multiserver-Job Stochastic Recurrence Equation
for Cloud Computing Performance Evaluation}

\author{Francois Baccelli}
\authornote{All authors contributed equally to this research.}
\email{Francois.Baccelli@ens.fr}
\orcid{0000-0002-9326-8422}
\affiliation{%
\institution{INRIA and Télécom Paris}
  \city{Paris}
  \country{France}}
\author{Diletta Olliaro}
\authornote{Most of this work was carried out while this author was affiliated with Università Ca' Foscari Venezia.}
\authornotemark[1]
\email{diletta.olliaro@unive.it, olliaro@kth.se}
\orcid{0000-0002-7361-819X}
\affiliation{%
  \city{Venezia}
\institution{Università~Ca'~Foscari~Venezia}
  \country{Italy}}
  \affiliation{
  \institution{KTH~Royal~Institute~of~Technology}
  \city{Stockholm}
  \country{Sweden}
  }
\author{Marco Ajmone Marsan}
\authornotemark[1]
\email{marco.ajmone@imdea.org}
\orcid{0000-0002-9560-7053}
\affiliation{%
\institution{IMDEA Networks Institute}
  \city{Leganes}
  \country{Spain}}
\author{Andrea Marin}
\authornotemark[1]
\email{andrea.marin@unive.it}
\orcid{0000-0002-5958-1204}
\affiliation{%
\institution{Università Ca' Foscari Venezia}
  \city{Venezia}
  \country{Italy}}

\renewcommand{\shortauthors}{Francois Baccelli, Diletta Olliaro, Marco Ajmone Marsan, and Andrea Marin}

\begin{abstract}

\normalsize
Cloud computing data centers handle highly variable workloads: job resource requirements can range from just one or few cores to thousands, and job service times can range from milliseconds to hours or days. This variability significantly limits the maximum achievable utilization of infrastructures.
Queuing theory has addressed the study of these systems with the definition of the Multiserver-Job Queuing Model (MJQM), where $s$ identical servers are present and  job $n$ requires $\alpha_n$ of the $s$ servers for a random amount of time $\sigma_n$. The $\alpha_n$
servers are occupied and released simultaneously. Unfortunately, despite its simple formulation, the MJQM remains elusive.
For example, the MJQM stability condition has been derived only in particular cases. As a consequence, even applying Discrete-Event Simulation (DES) under high load becomes challenging, because stability cannot be determined a priori.
In this paper, we analyze the MJQM with general independent arrival processes and service times under FCFS scheduling, using stochastic recurrence equations (SREs) and ergodic theory. 
Starting from the definition of the MJQM SRE, we prove the monotonicity and separability properties that allow us to apply an extension of Loynes' theorem, known as the monotone-separable framework, and formally define the MJQM stability condition. From these results, we introduce and implement two algorithms: the first one is used to draw sub-perfect samples (SPS) of the system's workload and the second one estimates the system's stability condition given the statistics of the jobs' input stream. 
The nature of the SPS algorithm allows for a massive GPU parallelization, thus greatly improving the efficiency in the estimation of performance metrics. The algorithm for the estimation of the stability condition solves an important problem for the analysis of MJQMs. 
We also define new metrics that capture the synchronization loss in MJQM systems and we show how these metrics can be efficiently evaluated using the SRE approach. 
Finally, we show that the approach proposed in this paper can be extended to more complicated systems, including MJQMs where resources have types.
\end{abstract}


\maketitle
\fancyfoot{}
\thispagestyle{empty}

\section{Introduction}

Cloud computing infrastructures (often called data centers) have become a key pillar of digital environments. Their number, size, complexity, and performance continue to grow, at the cost of massive and steadily increasing hardware deployments and energy consumption. Recent estimates show that, in the US, about $7$ million servers are delivered every year for data center usage, and their life expectancy is $5$ years only~\cite{berkeley}. Data centers offer both powerful computing resources and massive storage capacity, supporting the development of extremely large AI models now used across many domains, as well as the execution of less demanding tasks submitted by users for a wide range of purposes.
Unfortunately, due to an intrinsic high workload variability resulting from requests for variable types and quantities of resources, combined with the difficulty of the ensuing scheduling problem,
the utilization of this huge amount of deployed hardware is rather low. 
In other words, in order to be able to serve a highly unpredictable workload, many more resources are deployed than those actually needed, as a high percentage of resource waste is inevitable under simpler schedulers. This is quite different from traditional systems in which tasks usually differ only in their resource holding times, not in their resource demands. In these traditional cases, when the workload increases, utilization approaches asymptotically  $100\%$, as it is known from queuing theory.
Maximizing the practical utilization of data center cloud computing resources will reduce the hardware needed to handle workloads, leading to both economic and environmental benefits. 

The design and optimization of data centers and their schedulers are not yet supported by adequate quantitative tools. 
Typical questions that arise when analyzing data center performance are: What is the waiting time distribution for jobs with given characteristics? What is the maximum speed at which jobs can arrive in order not to saturate the data center? What is the maximum utilization reachable when jobs have some specific characteristics?
Both analytical and simulation models struggle to cope with the complexity of these systems.
Discrete event simulation is difficult because highly variable workloads, along with the complexity of the data center dynamics, require very long simulation runs to produce reliable estimates of performance metrics.
Analytical modeling is equally challenging because traditional tools, particularly queuing theory, cannot handle the complexity of such architectures and algorithms yet.

A new queuing theory construct that has recently been proposed to approach the analysis of data centers is the Multiserver-Job Queuing Model (MJQM)~\cite{MJQM}. The MJQM comprises an infinite capacity waiting line and a finite number of identical servers. Jobs (customers) arrive requesting a variable number of servers.
When the desired number of servers is available, the job starts its execution, and after a variable amount of time, when execution ends, the job releases the occupied servers and leaves. Jobs correspond to computing processes that must be executed. Servers represent the data center/cloud computer resources, that, in practice, are not indistinguishable, so that an extension of the MJQM to the case of multiple resource classes is necessary.
In roughly five years of research (the first papers on MJQM analysis appeared around 2020), not many results have been obtained by exploiting and extending traditional stochastic analysis tools.
Even the question of the stability condition for a general MJQM is still open. In fact, most results concern MJQM with only two classes of jobs, where the first class comprises jobs requesting a small constant number of servers (hence termed small jobs), while jobs of the second class require a number of servers which is a constant fraction (possibly all) of those available in the data center (big jobs). 

\paragraph{Our contribution} In this paper, we propose a new approach to the analysis of MJQMs with First-Come First-Served (FCFS) scheduling. We use a stochastic recurrence–based method, that is much more general than previously used tools for this class of questions in terms of the statistical assumptions that can be handled.
Early examples of successful applications of this method are in the analysis of data bases~\cite{BCM}, parallel processing of task graphs~\cite{BL} and more recently single or multiserver queues~\cite{blanchet0,blanchet}.
Within this framework, the MJQM state is defined by the workload at each server, representing the time an arriving job must wait before that server becomes available. This represents a significantly different viewpoint from that of the state-of-the-art modeling of MJQM, which typically focuses on server occupancy.
This method will be referred to as the Multiserver-Job Stochastic Recurrence Equation (MJSRE) below. It allows for multiple job classes, and general (not necessarily independent) inter-arrival and service times. It can be generalized to the case of multiple types of resources, i.e., of servers of different nature (that could, for example, model CPUs, GPUs, storage units, etc.) and thus better reflects the reality of data center/cloud computing dynamics.

The key observation on our approach is that the stochastic recurrence equation underlying MJQM with FCFS scheduling satisfies the assumptions of the so-called
Monotone-Separable framework~\cite{BF}. As a first consequence of the monotonicity properties, a proper extension of Loynes' theorem~\cite{BB} can be used to obtain perfect samples of the steady-state variables. Specifically, we introduce a novel algorithm that  
produces tight sub-perfect samples (SPS), namely lower bounds $L_n$ to perfect samples of the steady-state variables, indexed by a parameter $n$ linked to simulation time, and such that the sequence $L_n$ increases and couples in \textit{almost surely} (a.s.) finite time to the perfect sample.
Perfect samples are random variables with a distribution equal to that of the steady state of the stochastic process underlying our queuing model. This approach is quite useful because we can obtain independent and accurate estimates of perfect samples without explicitly knowing the stationary characteristics of the random process. Moreover, with a sufficient number of samples and under the assumption of Poisson arrivals, important system metrics can be estimated, together with rigorous confidence intervals, such as the utilization of servers or the distribution of the waiting times per job type. Finally, accurate approximations of perfect samples can be used as initial states of Discrete-Event Simulation (DES) runs that can provide estimates of time-dependent performance metrics, avoiding the problem of discarding the initial warm-up period.

The monotonicity and separability properties hold under very general assumptions on the arrival process, and this leads to a first characterization of the model stability region, i.e., we determine the values of the intensity of the arrival process for which the system is stable, given the number of servers and the characteristics of the arriving jobs. We show that this stability region is always of the threshold type on the arrival process intensity (note that this is not always true for queuing systems, as observed, for example in~\cite{Bramson}). The threshold can be represented as the growth rate of a pile, whose existence is guaranteed by the sub-additive ergodic theorem~\cite{BF}. While this is explicitly computable in some specific cases, we also propose a numerical algorithm to determine this threshold in very general scenarios, i.e., with general inter-arrival and service times and arbitrary classes of jobs.

Notice that, although we focus on MJQM with indistinguishable servers, at the end of the paper we discuss extensions to systems with multiple types of resources, or systems where jobs may require a specific resource.

In more detail, the main contributions of this paper are the following:
\begin{itemize}
    \item We apply the stochastic recurrence approach to the analysis of the FCFS MJQM in very general settings (in terms of number of servers requested by jobs, job inter-arrival and service time distributions).
    \item We prove monotonicity and separability of the multiserver-job stochastic recurrence equation.
    \item We use the saturation rule~\cite{BF} to characterize the stability region of the FCFS MJQM in terms of the job arrival rate, showing that the stability condition is of the threshold type.
    \item We use the Coupling From The Past approach, and more precisely Loynes' theory~\cite{loynes}, to devise an algorithm that produces tight SPS of the workload at each server.
    \item We derive sufficient conditions for moments of the stationary load in MJQM to be finite.
    \item We provide a representation of the stability threshold in terms of the growth rate of a  pile.
    \item We give algorithms for the computation of the stability threshold, the SPS, and the average number of wasted servers, i.e., those unused because of Head-of-Line (HOL) blocking.
    \item We discuss how the SPS algorithm and simulations based on forward execution of the MJSRE can be parallelized in the SIMD (single instruction multiple data) sense, to exploit the computational power of GPUs.
    \item We present extensive numerical results for various FCFS MJQM settings, validating the stochastic recurrence equation approach against exact results available in simple cases, and against DES results in more realistic cases; we also compare execution times in a systematic way.
    \item We use SPS to approximate waiting time averages, variances, percentiles, and distributions in the case of Poisson arrivals.
    \item We show how the SRE approach can be extended to multiserver-job queues where servers are of different types, and where job requirements are defined in terms of the number of servers required for each type.
\end{itemize}

The paper is organized as follows. 
Section~\ref{sec:related} reviews related work on the performance analysis of multiserver-job queuing systems. Section~\ref{sec:background} provides background information and introduces key concepts that support the analysis presented in this work. In Section~\ref{sec:mjsre}, we formally define the queuing model under investigation and present the related SRE that serves as the basis of our results. Section~\ref{sec:stability} contains the stability analysis. In Section~\ref{sec:metrics}, we focus on the existence of moments for the waiting time distribution and show how to compute key performance metrics. Section~\ref{sec:implementation} discusses our implementation of the MJSRE approach as well as the use of the MJSRE for SIMD parallel simulation. 
In Section~\ref{sec:numerical}, we present the experimental validation of our approach, highlighting the practical feasibility and precision of the proposed method. In Section~\ref{sec:generalizations}, we extend the proposed framework to more general settings, including systems with multiple resource types. Finally, Section~\ref{sec:conclusion} concludes the paper and outlines directions for future work.

\section{Prior Work}~\label{sec:related}
The study of the MJQM has attracted growing attention due to its relevance to modern data center operations~\cite{MJQM}, where jobs often span multiple servers and exhibit diverse resource and timing requirements. While many practical schedulers, such as SLURM~\cite{slurm}, Borg~\cite{borg}, and YARN~\cite{yarn}, {rely on carefully crafted heuristics based on data center management experience}, they lack rigorous analytical guarantees, particularly regarding system stability and performance metrics. 

Looking for analytical results, the FCFS policy represents a natural candidate to obtain a baseline against which more complex schedulers can be compared. In addition, FCFS offers the advantages of simplicity and fairness guarantees among different classes. However, even in the FCFS setting, early analytical studies of MJQM focused on highly constrained scenarios, for example assuming exponentially distributed service times that are identical across job classes, or considering a very limited number of available servers in the system~\cite{brill-1984, karatza, rumyantsev-2017, morozov, afanaseva-2020}.

More recent work is striving to relax these assumptions. In particular, the FCFS MJQM has been analyzed in scenarios with two job classes with exponentially distributed service times, yielding explicit expressions for the stability region~\cite{grosof-2023-mama, ourPEVA}, for systems with an arbitrary number of servers. 

These results were further extended to more general 2-stage Coxian service time distributions (still for two classes), where the stability region is expressed in matrix form~\cite{anggraito-mascots-2024}. Similar matrix-analytic methods have also been applied to compute performance metrics under both exponential and Coxian assumptions~\cite{anggraito-mascots-2024, anggraito-cox-2025}. 
More recently,~\cite{grosof-2024-marcreset} derived explicit bounds on mean response times, allowing for an arbitrary number of classes. 

With respect to the specific methodology proposed in this work, a natural question is whether perfect samples can be derived in place of sub-perfect ones as defined above. For instance, the studies presented in~\cite{blanchet, blanchet0} introduced a coupling from the past technique for perfect sampling in GI/GI/$c$ queues (which is a special case of our model), based on a computable stationary majorant derived from a random assignment (RA) queuing system. As discussed in Appendix~\ref{app:sandwich}, this framework cannot be directly applied to the FCFS MJQM. The reason is that, in the more general context of the MJSRE, the corresponding RA system exhibits a more restrictive stability region than the initial model, so that the RA system upper-bound cannot be used in general.
{The identification of a stability region-preserving majorant for the MJQM is not straightforward, and will be considered in future works.}

Beyond the studies mentioned above, to the best of our knowledge, there are no general results for FCFS MJQM, especially for arbitrary job class counts and non-exponential service distributions. This motivates the use of alternative analytical tools, such as the stochastic recurrence approach, that can characterize the stability region for general FCFS MJQM settings.

Some performance studies of non-FCFS MJQM have also appeared in the recent literature~\cite{grosof-2023-serverfilling, maguluri-infocom-2012, bichler, beloglazov}, but are less related to this work.

\section{Background}~\label{sec:background}
This section describes the background for some of the key concepts used in this paper.

\subsection{Stochastic Recurrence Equations}
According to~\cite{BB}, stochastic recurrence equations are equations of the form $x_{n+1}=f(x_n,u_n)$ where $x_n$ belongs to some measurable space $({\mathcal X},{\mathcal F})$, the sequence $\{u_n\}$ is assumed to be stationary and ergodic, with $u_n$ belonging to some measurable space $({\mathcal U},{\mathcal G})$, and $f$ is a measurable function from ${\mathcal X} \times {\mathcal U}$ to $\mathcal X$.
The simplest queuing example is Lindley's equation for the workload in a G/G/1 queue: $W_{n+1} = (W_n + \sigma_n - \tau_n)^+$, where $W_n\in \mathbb R^+$ is
the workload at the arrival of job $n$, $\sigma_n \in \mathbb R^+$ is the service requirement of job $n$ and $\tau_n \in \mathbb R^+$ is the time between the arrival of job $n$ and that of job $n+1$.
This representation of a queuing problem in terms of such an equation paves the way to analyze this problem in terms of a dynamical system, and more precisely to analyze its stationary regime in terms of the theory of measure-preserving transformations.

\subsection{Loynes' Theorem}
Loynes~\cite{loynes} was the first to study Lindley's equation using the theory of dynamical systems. He proved that coupling from the past starting from an empty state at time $-n$ leads to a state at time 0 that increases with $n$. He proved that, if the sequence ${u_n}$ is ergodic, either the limit as $n\to \infty$ is infinite \textit{almost surely} (a.s.), which corresponds to the unstable case (load factor higher than 1) or the limit is a.s. finite, which corresponds to the stable case (load factor less than 1). In the latter case, the limit is a perfect sample of the steady-state workload distribution (see~\cite{BB}).
A perfect sample of a distribution $F$ is any random variable $X$ with distribution equal to $F$.
Assume that the state space of $X$ is equipped with a partial order $\le$. Then a sub-perfect sample is lower bound to
a perfect sample $X$, namely any other random variable $Y$ such that $Y\le X$ a.s.

\subsection{Monotone-Separable Stochastic Networks}

Consider a stochastic network described by the following framework:
\begin{itemize}
\item The network has a single input point process $N$, with points $\{T_n\}$; for all $m\leq n\in \mathbb N$; let $N_{[m,n]}$ be the $[m,n]$ restriction of $N$, namely the point process with points $\{T_l\}_{m\leq l\leq n}$.
\item The network has finite activity for all finite restrictions of $N$: for all $m\leq n\in \mathbb N$, let $X_{[m,n]} (N)$ be the time of the last activity in the network, when this one starts empty and is fed by $N_{[m,n]}$. We assume that for all finite $m$ and $n$ as above, $X_{[m,n]}$ is finite.
\end{itemize}
Assume that there exists a set of functions $\{ f_l \}$, $f_l:\mathbb R^l\times K^l\to \mathbb R$, such that:
$$ X_{[m,n]}(N) = f_{n-m+1} \{ (T_l, {\xi}_l),\ m \leq l \leq n \}, \label{eq0} $$
for all $n,m$ and $N$, where the sequence $\{\xi_n\}$ is that of service times and routing decisions.

We say that the network described above is monotone--separable if the functions $f_n$ are such that the following properties hold for all $N$:
\begin{description}
\item[1.] {\bf (causality):}
for all $m \leq n$, $X_{[m,n]}(N) \geq T_n$;
\item[2.] {\bf (external monotonicity):}
for all $m \leq n$, $X_{[m,n]}({{N'}}) \geq  X_{[m,n]}(N),$
whenever $N'= \{T'_n\}$ is such that $T'_n \geq T_n$ for all $n$, a property which we will write $N'\geq N$ for short;
\item[3.] {\bf (homogeneity):}
for all $c\in \mathbb R$ and for all $m \leq n$, $X_{[m,n]}(c+N) = X_{[m,n]}(N) + c;$
\item [4.] {\bf (separability):}
for all $m \leq l < n$, $X_{[m,l]}(N) \leq T_{l+1}$, then
$X_{[m,n]}(N) = X_{[l+1,n]}(N).$
\end{description}

\section{The Multiserver-Job Stochastic Recurrence Equation}\label{sec:mjsre}

In this section, we formally present the queuing model under consideration along with the associated notation. We then introduce the stochastic recurrence equation used for its analysis.

\subsection{Model Description}
Consider a queuing system consisting of $s$ identical servers and one infinite-capacity waiting room. Each arriving job $n$, $n\in \mathbb Z$, is characterized by a resource demand $\alpha_n$, i.e., the number of servers required, and a service time $\sigma_n$. A job $n$ arriving at epoch $t_n$ remains in the waiting room until it reaches the head of the waiting line and there are at least $\alpha_n$ available servers. We assume a First-Come First-Served (FCFS) discipline. 
All the servers required by a job are allocated and de-allocated simultaneously at the instant at which the job enters in service and when it departs, respectively. Notice that, as a consequence, at the departure time of a job, the scheduler puts in service all the jobs in the buffer, scanning them according to their arrival order, until it finds the first whose request cannot be satisfied, or none remains.
It may happen that, while job $n$ is at the head of line, a subsequent job $m$ with resource demand $\alpha_m<\alpha_n$ could start its service since there are at least $\alpha_m$ idle servers, but is blocked by the presence of job $n$. This situation is known as \emph{Head-of-Line blocking} (HOL) and is depicted in Figure~\ref{fig:hol}. 

In modeling data centers, the available CPUs (or cores) correspond to the $s$ servers in the proposed queuing system. HOL is an important factor that contributes to poor resource utilization in MJQM, particularly so in the presence of large differences in the number of servers required by jobs. This issue is particularly severe under FCFS scheduling. Designing a non-preemptive and job-agnostic scheduling discipline that balances high utilization, fairness, and low delays remains a challenging task.

\begin{remark}[Why is FCFS scheduling important in cloud computing data center performance analysis?]
Data centers are managed by complicated policies that may allow for job preemption, or require jobs to declare their service times as in the case of Backfilling~\cite{feitelson-2001}. Then, why is the analysis of FCFS policy attracting much attention (see, e.g.,~\cite{grosof-2023-mama, grosof-2024-marcreset, anggraito-mascots-2024, anggraito-cox-2025, grosof-2024-marcreset, ourPEVA, afanaseva-2020, brill-1984, ourTPDS})? Beside representing the fundamental case for many queuing systems, FCFS is the simplest non-preemptive policy that allows for a mathematical understanding of the under-utilization of data center resources. FCFS serves as  baseline to compare the performance of more complex scheduling policies~\cite{bichler, beloglazov, grosof-2023-serverfilling, grosof-2022-pomacs} in terms of response time, maximum throughput and fairness (recall that FCFS treats all jobs equally). Therefore, a mathematical understanding of the behavior of systems controlled by FCFS schedulers is the key to study the trade-offs introduced by other, and more sophisticated scheduling strategies.
\end{remark}

\begin{figure}[!htbp]
\centering
\includegraphics[width=0.35\textwidth]{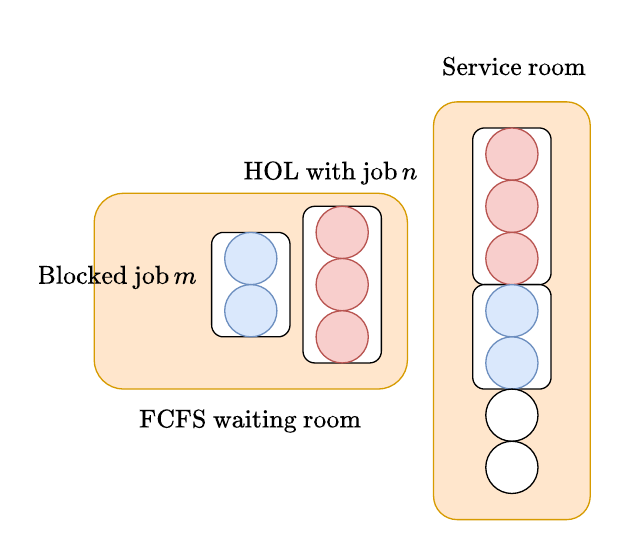}
\caption{Example of HOL blocking in Multiserver-job systems.}\label{fig:hol}
\Description{Example of HOL blocking in Multiserver-job systems.}
\end{figure}

\subsection{The Kiefer and Wolfowitz Stochastic Recurrence Equation}\label{sec:multiserver}
We first recall the stochastic recurrence equation for the multiserver queue where each job requires exactly one server~\cite{BB}. Jobs are characterized by their arrival time $t_n$ and service time $\sigma_n$.
We denote the inter-arrival time between job $n$ and job $(n+1)$ as $\tau_n$, i.e., $\tau_n=t_{n+1}-t_n$. 
The system has $s \geq 1$ servers that attend jobs,
and the scheduling policy is FCFS. 

According to the Kiefer and Wolfowitz model, the state of the system is an ordered vector of non-negative real components that represent the workload at each server. The FCFS discipline corresponds to the allocation rule that assigns an arriving job to the server with the smallest workload. Once assigned, this job will wait and will then be served at unit rate until completion.

The ordered workload vector process $\{ W_n \}, n \in \mathbb N$, where $W_n = (W_n^1,\ldots, W_n^s)\in {\mathbb{R}}_{\geq 0}^s$, is defined as a permutation in increasing order of the workload found in each server by the $n$-th job upon arrival: $W_n^1 \leq W_n^2 \leq \cdots \leq W_n^s,$ for all $n \in  \mathbb N $.

This ordered vector satisfies the Kiefer and Wolfowitz SRE:
\begin{equation}\label{eq.1}
W_{n+1} = {\mathcal R}( W_n + \sigma_n  \load - \tau_n  \un)^+\,, 
\end{equation}
where $ \load = E_1= (1,0,\ldots,0),\  \un = (1,\ldots,1)$, ${\mathcal R}$ is the operator arranging vectors of $\mathbb R_{\geq 0}^s$ in increasing order, and $(x^1,\ldots,x^s)^+ = ((x^1)^+,\ldots,(x^s)^+)$, with $x^+= \text{max}(x, 0)$. 

\subsection{The Multiserver-Job Stochastic Recurrence Equation}

\subsubsection{Definition of the stochastic recurrence equation and monotonicity}
For the analysis of the multiserver-job queuing model, we follow the lines presented for the Kiefer and Wolfowitz stochastic recurrence equation, but we introduce some important extensions. 

Similarly to Subsection~\ref{sec:multiserver}, the state just before time $t_n$ is the ordered workload vector $W_n\in \mathbb{R}_{\geq 0}^s$. 
In the FCFS case considered here, to start the service of job $n$, it is necessary to wait for the availability of the
latest of the first $\alpha_n$ CPUs, which occurs at $W_n^{\alpha_n}$.
Hence, the ordered workload vector satisfies the Multiserver-Job Stochastic Recurrence Equation
(MJSRE):
\begin{equation}\label{eq.2}
	W_{n+1} = {\mathcal R}( \sync(\alpha_n, W_n) + \sigma_n  \load(\alpha_n) - \tau_n  \un)^+\,,
\end{equation}
where $\load(\alpha)\in \mathbb{R}^s$ is defined by $\load(\alpha)^i= 1_{i\le \alpha}$ for all $1\le i\le s$, or equivalently $\load(\alpha)=E_1+\cdots+E_\alpha$ (with $\{E_i, i=1, \dots, s\}$ the orthogonal basis of $\mathbb{R}^s$.) and, for all ordered $W$, $\sync(\alpha, W)\in \mathbb{R}^s$ is defined by:
\[
\sync(\alpha, W)^i = W^{\alpha}\vee W^i,\quad \mbox{for all $1\le i\le s$}\,,
\]
with $a\vee b: =\max(a,b)$.
The addition of $\sigma_n  \load(\alpha_n)$ is because the load $\sigma_n$ is added to the workload of the $\alpha_n$ least loaded CPUs. That of $- \tau_n  \un$ is because the service provided up to $\tau_n$ is subtracted from the workload of each CPU.
The first term $\sync(\alpha_n, W_n)$ is necessary because the synchronous services of the required CPUs must be simultaneous and can only start when CPU $\alpha_n$ becomes free. 
Notice that the HOL effect is caught by this term, since the first $1, \ldots, \alpha_n-1$ servers remain idle until server $\alpha_n$ becomes available for the service of the job at the head-of-line.

We now introduce an important lemma that will be crucial for the development of our analysis. 

\begin{lemma}\label{lemma:monotonic}
	For all fixed $\alpha\in [1,\ldots,s]$, $\sigma$ and $\tau$ in $\mathbb R^+$, the map:
\begin{equation}\label{eq.21}
	W \to {\mathcal R}( \sync(\alpha, W) + \sigma  \load(\alpha) - \tau  \un)^+\,,
\end{equation}
is coordinatewise non decreasing. That is, for all ordered workload vectors $W$ and $W'$ of  $\;\mathbb{R}_{\geq 0}^s$ such that coordinatewise $W\leq W'$, it holds that, coordinatewise:
\[
{\mathcal R}(\sync(\alpha,W)+\sigma \load(\alpha)-\tau \un)^+ \leq \mathcal R(\sync(\alpha,W')+\sigma \load(\alpha)-\tau \un)^+\,.
\]
\end{lemma}
\emph{Proof.} 
By first principles, one first gets that the map $W \to \sync(\alpha, W)$ is monotonic in the sense defined above.
It immediately follows that the map $W \to (\sync(\alpha, W) + \sigma  \load(\alpha) - \tau  \un)^+$ is also monotonic.
Now, the proof will be completed by recalling that the map from a vector to its reordering is isotonic.  To see this, consider a vector $A$ and $A+aE_k$ for some index $k$ and $\alpha>0$. Clearly, $A\leq A+ \alpha E_k \implies \mathcal R(A) \leq \mathcal R(A+a E_k)$.

This concludes the proof. \qed

\begin{theorem}\label{th:ms}
The multiserver-job queuing system is monotone-separable.
\end{theorem}

\begin{proof}
Here, we have: $X_{[m,n]}=T_n+\max(W_{[m,n]}^{\alpha_n}+\sigma_n,W_{[m,n]}^s)$,
where $W_{[m,n]}$ is obtained from the recurrence equation:
$$W_{[m,m+k+1]} = {\mathcal R}(S(\alpha_{m+k},W_{[m,m+k]}) + \sigma_{m+k}  L(\alpha_{m+k}) - \tau_{m+k}  \un)^+,\quad k\ge 0\,,  $$
with $W_{[m,m]}=0$.
The causality property follows from the fact that workload is non-negative.
The homogeneity property follows from the fact that $\{W_{[m,m+k]}\}_k$ is left unchanged by a
translation of the arrival process.
The separability property follows from the fact that,
if $T_{l+1}\ge X{[m,l]}$, then
$W_{[m,l+1]}=0$, so that $W_{[m,n]}=W_{[l+1,n]}$ for all $n\ge l+1$.

We now prove the external monotonicity property.
Let $D_{[m,n]}$ be the vector: $D_{[m,n]}:= T_n\un+W_{[m,m+n]}$. We will prove the property that for all $n$, $D_{[m,n]}(N)\le D_{[m,n]}(N')$.
This will imply the desired monotonicity property.
The proof is by induction.
The base case holds as $D_{[m,m]}(N)=T_m\un\le T'_m\un =D_{[m,m]}(N')$. Assume now, by the induction assumption, that $D_{[m,n]}(N) \le D_{[m,n]}(N')$, i.e.: $T_{m+n}\un+W_{[m,m+n]}(N)\le T'_{m+n}\un +W_{[m,m+n]}(N')$.
Then:
        $$W_{[m,m+n]}(N) \le W_{[m,m+n]}(N') +(T'_{m+n}-T_{m+n})\un\,. $$
Hence, by Lemma~\ref{lemma:monotonic}:
        \begin{eqnarray*}
                W_{[m,m+n+1]}(N) & \le & 
                {\mathcal R}(S(\alpha_{m+n},W_{[m,m+n](N')}) + \sigma_{m+n}  L(\alpha_{m+n})  +(T'_{m+n}-T_{m+n} - \tau_{m+n})\un)^+\\
                & = &
                {\mathcal R}(S(\alpha_{m+n},W_{[m,m+n](N')}) + \sigma_{m+n}  L(\alpha_{m+n})  +(T'_{m+n+1} -T_{m+n+1})\un -
                \tau'_{m+n})\un)^+\\
                & \le &
                {\mathcal R}(S(\alpha_{m+n},W_{[m,m+n](N')}) + \sigma_{m+n}  L(\alpha_{m+n}) -
                \tau'_{m+n})\un)^+ +(T'_{m+n+1} -T_{m+n+1})\un \\
                & = &
                W_{[m,m+n+1]}(N') + (T'_{m+n+1} -T_{m+n+1})\un\,,
        \end{eqnarray*}
        where we used the fact that if $b>0$, then $(a+b)^+\le a^+ +b$.
This proves that $D_{[m,m+n+1]}(N) \le D_{[m,m+n+1]}(N')$, which concludes the induction.
\end{proof}

\subsubsection{The associated Loynes' sequence}
We are now in the position to define the Loynes' sequence associated with the MJSRE.
For this, it is best to introduce the following (left) shift on bi-infinite sequences $u=\{u_n\}_{n\in \mathbb Z}$:
$$\theta \left(\ldots, \newtop{u_{-2}}{{-2}}, \newtop{u_{-1}}{{-1}},\newtop{{u_0}}{0},\newtop{{u_1}}{1},\newtop{{u_2}}{2},\ldots\right)
=\left(\ldots, \newtop{u_{-1}}{-2},\newtop{u_{0}}{{-1}},\newtop{u_1}{0},\newtop{u_2}{1},\newtop{u_3}{2},\ldots\right)\,,$$
where the lower line in the array gives the position of the entry in the sequence.
Note that $\theta$ acts on the stationary sequence $\xi=\{\xi_k\}_{k\in \mathbb Z}$ with $\xi_k=(\alpha_k,\sigma_k,\tau_k)$.
For instance, given the bi-infinite sequence $\{\alpha_n\}_{n\in \mathbb Z}$ with $\alpha_0$ in position 0,
$\alpha_0 (\theta (\alpha))= \alpha_1$ and more generally, for all $n\in \mathbb Z$,
$\alpha_0 (\theta^n (\alpha))= \alpha_n$.
In addition, for all $n\ge 0$, $W_n$ defined above is a function of $\xi$.

Now, the Loynes sequence $\{M_n\}_{n\in \mathbf N}$ associated with the MJSRE is also a function of $\xi$ defined by induction
by $M_0=0 \in \mathbb{R}^s$ and:
\begin{equation}\label{eq.21L}
	M_{n+1}\circ \theta = {\mathcal R}( \sync(\alpha_0, M_n) + \sigma_0  \load(\alpha_0) - \tau_0  \un)^+,\quad
	n\ge 0\,.
\end{equation}
For instance,
$M_{1}= {\mathcal R}( \sync(\alpha_{-1}, 0) + \sigma_{-1}  \load(\alpha_{-1}) - \tau_{-1}  \un)^+$
is the ordered workload at time 0 if the system starts empty at time $-\tau_{-1}$, when using as CPU set $\alpha_{-1}$, as load $\sigma_{-1}$ and as next inter-arrival $\tau_{-1}$ at that time. It is a function of $\xi_{-1}$ only, whereas $M_{2}= {\mathcal R}( \sync(\alpha_{-1}, M_1(\theta^{-1}(\xi)) + \sigma_{-1}  \load(\alpha_{-1}) - \tau_{-1}  \un)^+$ is easily seen to be the ordered workload at time 0 if the system starts empty at time $-\tau_{-1}-\tau_{-2}$, when using as CPU set $\alpha_{-2}$, as load $\sigma_{-2}$ and as next inter-arrival $\tau_{-2}$ at that time, and with the very same as above at time $-\tau_{-1}$, namely a function of $\xi_{-2}$ and $\xi_{-1}$ only. 
More generally, $M_n$ is the state of the system at time 0 when the system starts empty at time $-(\tau_{-n}+\cdots+\tau_{-1})$ when using the data set $(\xi_{-n},\cdots,\xi_{-1})$ at the arrival epochs $-(\tau_{-n}+\cdots+\tau_{-1}),\ldots,-\tau_{-1}$.

The following result is immediate by induction.

\begin{lemma}\label{lemma:convergence}
The sequence $M_n$ is monotonically non-decreasing in $n$.
\end{lemma}
Hence, there exists an almost sure limit $M_\infty$ of $M_n$ as $n$ tends to $\infty$.
In addition, since $\theta$ preserves the law of $\{\xi_n\}_{n\in \mathbb Z}$ (we assume that this sequence is stationary), it is easy to see that for all $n\ge 0$, $M_n$ has the same law as $W_n$ when taking $W_0=0$.
As a consequence, we can state the following lemma, which contains two important properties that allow us to introduce the algorithm for the computation of a tight SPS of the workload for the multiserver-job queuing system.  
\begin{lemma}\label{lemma:perfect}
The following properties hold:
\begin{itemize}
\item $W_n$ converges weakly (from below) to $M_\infty$ when $n$ tends to $\infty$;
\item $M_\infty$ provides a perfect sample of the stationary distribution of the state vector of the MJSRE.
\end{itemize}
\end{lemma}

\subsection{The sub-perfect sampling algorithm}\label{subsec:perfalg}

Lemma~\ref{lemma:convergence} and Lemma~\ref{lemma:perfect} allow us to define an algorithm  to draw tight SPS of the steady state of the MJSRE.

The algorithm, which implements Equation~(\ref{eq.21L}), receives  integer $\ell$ in input. It computes the maximum workload of the sequence from $-\ell$ to $0$ and the maximum workload of the sequence from $-2\ell$ to $0$. If the two maxima are the same, we consider that the algorithm has converged; otherwise, we double $\ell$ and repeat the iteration. 
We will discuss the choice of the value of $\ell$ in Section~\ref{sec:implementation}.

\begin{algorithm}[tbp]\footnotesize
\caption{SPS (Sub-Perfect Sample)}
\label{alg:general}
\begin{algorithmic}[1]
\Procedure{PerfectSampleLB}{$\ell$}
\State $\ell^{prev} \gets 0$ \Comment{Previous $\ell$}
\State $M^{max}\gets (-\infty,\ldots,-\infty)$ \Comment{Maximum workload}
\State $\text{Convergence} \gets$ \textbf{false} \Comment{Stop condition}
\Do
	\For {$n \gets \ell^{prev}+1$ to $\ell$}
       \State $\tau_n \gets \Call{SampleInterArrival}{n}$ \Comment{Sample time $\tau_{-n}$}
        \State $(\alpha_{n}, \sigma_{n})\gets \Call{SampleJob}{n}$ \Comment{Sample job $-n$}
    \EndFor
       {\State $M \gets (0,\ldots,0)$
       \For {$j \gets \ell $ downto $0$}
       \State $M= \mathcal R(S_{\alpha_{j}}(M)+\sigma_{j}e_{\alpha_{j}}-\tau_j e)^+$ 
	\EndFor}
    \State $M^*=M$
    \If {$M^*>M^{max}$}
       \State $M^{max}=M^*$
       \State $\ell \gets \ell * 2$
	{ \State $\ell^{prev} \gets \ell $}
    \Else
       \State Convergence $\gets$ \textbf{true}
    \EndIf
\doWhile{\textbf{not} Convergence}
\State \Return $M^{max}$
\EndProcedure
\end{algorithmic}
\end{algorithm}

Here are some observations on Algorithm~\ref{alg:general}:
\begin{enumerate}
\item $M$ is the workload generated by the sequence that begins with job $-\ell$ and terminates at the arrival time (but excluding) 
job $0$, therefore accounting for the services from $\sigma_{-\ell}$ to $\sigma_{-1}$ at any step.
\item $M^{max}$ is the maximum workload observed at any step.
\item It follows from the last two lemmas and from the results presented in~\cite[Ch. 2]{BB} that, if the MJSRE is stable, then
\begin{enumerate}
\item Algorithm~\ref{alg:general} \textbf{terminates almost surely} (this follows from the fact that $M_n$ is equal to $M_{\infty}$ for all $n\ge N$ for some finite $N$);
\item $M^{max}$ is a {\bf{lower bound to a perfect sample}} (i.e., a sub-perfect sample -- SPS) of the system workload;
\item $M^{max}$ tends to the perfect sample $M_{\infty}$ as $\ell \rightarrow \infty$.
\end{enumerate}
\end{enumerate}

The following remark will be used in Section~\ref{sec:implementation}, where we discuss the implementation of the algorithm.

\begin{remark}[Subsequences]\label{rem:subseqeunces}
Let us extend our notation as follows. Let $n_1>n_2\ge 0$ and consider the sequence of arrivals from job $-n_1$ to job $-n_2$. Let the initial workload at the arrival $-n_1$ be $M^{\dagger}$.
Denote the workload at the arrival of job $-n_2$ under this assumptions by $M_{n_1,n_2}(M^{\dagger})$. Then, the following holds true:
$
M_{n_1,0}(0)=M_{n_2,0}( M_{n_1,{n_2}}(0))\,.
$
\end{remark}

\begin{remark}[Comparison with a forward approach to MJSRE and ordinary DES] 
Algorithm~\ref{alg:general} produces a sub-perfect sample of the steady-state variables for any finite choice of $\ell$. If $\ell=\infty$, by Loynes' theorem, we obtain a perfect sample. The key-point of Algorithm~\ref{alg:general} is that it introduces a stopping criterion for the computation (independent of the load or of other characteristics of the system) that guarantees to produce a SPS as a result. Clearly, one could also use the MJSRE in a forward manner, starting from the empty system and carrying out a MJSRE-based simulation. We call this approach forward MJSRE (F-MJSRE) and we will use B-MJSRE (backward MJSRE) to refer to Algorithm~\ref{alg:general} when the context requires it.  If the sequence is long enough, we expect to obtain a good approximation of a perfect sample. However, the forward approach does not give statistical guarantees about the fact that the number of processed jobs is sufficient to have an accurate approximation of the perfect sample and this is experimentally discussed in Section~\ref{sec:implementation}. A totally different approach to the simulation of MJQM is what, in this paper, is called DES. This consists in the execution of a simulation based on maintaining a future event list containing the departure times of all jobs in service and the next arrival time. Since the nature of the $n$-th event (e.g., arrival, departure) depends on the samples of the random variables (e.g., service times, inter-arrival times) needed to determine this event, DES can only be parallelized in a straightforward manner by distributing simulation runs across different CPUs, not
by using the SIMD/GPU type parallelism allowed by the SRE approach.
We will discuss in more detail in Section~\ref{sec:implementation} how F-MJSRE and B-MJSRE are massively parallelizable thanks to the use of GPUs, while this is not the case, at the state-of-the-art, for DES.

\end{remark}

\section{Stability Analysis}\label{sec:stability}
In this section, we assume that the service times $\sigma_n$ have a finite first moment.
The stochastic system under consideration belongs to the Monotone-Separable framework as shown in Theorem~\ref{th:ms}. This leads to the following theorem, which shows two important things: the stability is of the threshold type on the job arrival rate; the threshold in question admits a representation in terms of 
the growth rate of a   pile. 
\begin{theorem}\label{th:stability}
The stability condition (the condition under which the workload vector $W_n$ converges weakly to a finite limit as $n\to \infty$) is of the form: 
\begin{equation} 
\lambda <\lambda_c,\quad \text{where }\lambda_c= 1 / \gamma\,,
\end{equation}
where $\gamma$ is the constant:
\begin{equation}\label{eq.2a}
\gamma:= \lim_{n\to \infty} |H_n|_\infty / n \,,\quad a.s.
\end{equation}
Here, $\{H_n\}$ is the  pile vector given by $H_0=0$ and the recurrence equation:
\begin{equation}\label{eq.2h}
H_{n+1} = {\mathcal R}( \sync(\alpha_n, H_n) + \sigma_n  \load(\alpha_n) )\,,
\end{equation}
describes the evolution of workload with respect to $n$ when $\tau_k=0$ for all $k$, that is, when all jobs are immediately available at the beginning of time.
\end{theorem}
\begin{proof}
The statement is a rephrasing of the saturation rule in~\cite{BF}.
As proved in~\cite{BF}, it follows from the sub-additive ergodic theorem that the   pile $H_n$ has a sup norm $|H_n|_\infty$ with a deterministic growth rate $\gamma$, which justifies the existence of the a.s. limit in Equation~(\ref{eq.2a}). 
\end{proof}
Here are a few observations. There are some situations of practical relevance where $\gamma$ can be evaluated explicitly. One of them is that where there are jobs which require all servers. Denote by $p_g$ the frequency of such jobs (which will be referred to as of type $g$ - for global).  In this case, it is immediate (from the cycle formula of the theory of regenerative processes) that:
\begin{equation}\label{eq.2n}
        \gamma = p_g \mathbb{E} [H^g_\eta]\,,
\end{equation}
where $H^g$ is the   pile defined above when the first job is global and $\eta$ is the smallest $n>1$ such that job $n$ is of type $g$. Note that the distribution of $\eta$ is geometric of parameter $p_g$.
It is worth noting that this result reduces to Equation (6) in~\cite{ourPEVA}, that, together with Equation (3) in the same paper, gives an explicit expression for the stability condition in the case of jobs requiring either all servers or just 1, with the latter having exponentially distributed service times. In the case of deterministic service times of small jobs, Equation (\ref{eq.2n}) is equivalent to Equations (1) and (3) in~\cite{ourTPDS}.

This condition can be extended to other and less constraining scenarios.  In the general case, the complexity of the evaluation of $\gamma$ is an open question. For this reason, it can be useful to evaluate the limit in Equation~(\ref{eq.2a}) numerically as shown in Algorithm~\ref{alg:limit}.

\begin{remark}[Applicability of Algorithm~\ref{alg:limit}]
It is worth highlighting that, thanks to the monotonicity and separability properties, Algorithm~\ref{alg:limit} is, to the best of our knowledge, the first general result for the numerical evaluation of the stability region of the FCFS MJQM. This result holds for an arbitrary number of servers and job classes, and accommodates general service time distributions. Moreover, it is important to notice that the stability region provided by the algorithm is insensitive to the distribution of inter-arrival times.
\end{remark}

\begin{algorithm}[tbp]\footnotesize
\caption{Stability Region}
\label{alg:limit}
\begin{algorithmic}[1]
\Procedure{Stability}{$\ell$, $\varepsilon$}
\State $\ell^{prev} \gets 0$ \Comment{Previous $\ell$}
\State $\gamma^{cur} \gets 0$ 
\State $\gamma^{prev} \gets -\infty$ 
\State $M \gets (0,\ldots,0)$
\While{$|\gamma^{cur}-\gamma^{prev}|>\varepsilon$}
    \State $\gamma^{prev} \gets \gamma^{cur}$
	\For {$n \gets \ell^{prev}+1$ to $\ell$}
        \State $(\alpha_{n}, \sigma_{n})\gets \Call{SampleJob}{n}$ \Comment{Sample job $n$}
    \EndFor       
       \For {$n \gets \ell^{prev}+1 $ to $\ell$}
       \State $M= \mathcal R(S_{\alpha_{n}}(M)+\sigma_{n}e_{\alpha_{n}})$ \Comment{Compute the   pile}
	\EndFor
    \State $\gamma^{cur} \gets M^s/\ell^{cur}$
    \Comment{Approximate the limit}
    \State $\ell^{prev}\gets \ell^{cur}$
\EndWhile
\State \Return $1/\gamma^{cur}$
\EndProcedure
\end{algorithmic}
\end{algorithm}
\section{Metrics}\label{sec:metrics}
For all the metrics defined in what follows, it is assumed that the system is stable in the sense
defined above. The focus is on jobs executed within a cloud computing data center. In what follows, we use $W$ to denote the limiting value of $M_\infty$ according to Lemma~\ref{lemma:convergence}.
\subsection{Moments}\label{sec:moments}
A natural question (which is important in the evaluation of confidence intervals)
is that of the finiteness of moments of order $k$ (e.g., of order 2) for the maximal
coordinate of $W$, namely $W^s$.
This question was studied in detail in~\cite{BF2}, under the general framework of
monotone-separable networks. It directly follows from that paper that,
in the Poisson case for instance, if $\sigma_0$ has finite moments of order $k+1$,
then $W^s$ has finite moments of order $k$.
\subsection{Job Observer Metrics}
Under the stability condition, the backward construction described in Section~\ref{sec:mjsre}
provides stationary samples of the ordered waiting time process $W\in \mathbb{R}^s$. This
stationarity is w.r.t. the discrete shift from one arrival to the next, and
the distribution that this shift preserves is the Palm probability w.r.t. the
arrival point process.

We discuss here metrics that are meaningful for jobs, or, equivalently, under the Palm
distribution of the (job) arrival point process.
This includes the stationary waiting time and the system time (or execution time) of a job.
The former, which is the duration between the arrival time of job 0 ($t_0=0$) and
the time all servers it asks for are available, is equal to $W^{\alpha_0}$,
where $\alpha_0$ is the number of servers that job 0 asks for.
The latter is equal to $r=W^{\alpha_0}+\sigma_0$.
Hence, the samples of $W$ allow one to provide samples of both the stationary waiting and the stationary system time.
These samples in turn allow one to get estimates and confidence intervals of e.g. moments of $v$ or $r$ by
using the strong law of large numbers and the central limit theorem, respectively.

\subsection{System Level Metrics}

The Palm inversion formula allows one to evaluate continuous time metrics of interest to
system level analysis. Here is a typical question: what is the mean number of processors active
in steady state?
This will, for instance, determine the energy consumption of the owner of the data center.

On coordinate $i$, the server is busy by load brought by jobs until (and including) job 0 until time 
$W^{\alpha_0} + \sigma_0$, for $1\le i\le \alpha_0$, and until time $W^i$, for $\alpha_0< i\le s$.

More precisely, in the absence of further arrivals, the servers are busy or idle as follows. For $1\le i\le \alpha_0$, server $i$ is:
\begin{eqnarray}
        \begin{cases}
                \mbox{busy} & \mbox{on } [0 , W^{i}]  \\
                \mbox{idle (due to HOL blocking)} & \mbox{on } [W^{i},W^{\alpha_0}]       \\
                \mbox{busy} & \mbox{on } [W^{\alpha_0},W^{\alpha_0}+\sigma_0]         \\
                \mbox{idle} & \mbox{on } [W^{\alpha_0}+\sigma_0,\infty]\,,
        \end{cases}
\end{eqnarray}
whereas for $\alpha_0< i\le s$, server $i$ is:
\begin{eqnarray}
        \begin{cases}
                \mbox{busy} & \mbox{on } [0 , W^{i}]\\
                \mbox{idle} & \mbox{on } [W^{i},\infty]\,.
        \end{cases}
\end{eqnarray}

Hence, in the time interval $[0,\tau]$, the total server idle time is:
\begin{equation}
\sum_{i=1}^{\alpha_0} (W^{\alpha_0}\wedge \tau-W^{i}\wedge \tau + \tau- (W^{\alpha_0}+\sigma_0)\wedge \tau)
+ \sum_{i=\alpha_0+1}^s (\tau-W^i\wedge \tau)\,,
\end{equation}
with $a\wedge b: =\min(a,b)$.

It follows from the inversion formula of Palm calculus~\cite{BB} that the mean number of servers inactive in the continuous time stationary regime of the system (or total steady-state waste) is: 
\begin{equation}
w=\lambda \mathbb{E}_0[ \sum_{i=1}^{\alpha_0} (W^{\alpha_0}\wedge \tau-W^{i}\wedge \tau + \tau- W^{\alpha_0}+\sigma_0)\wedge \tau)
+ \sum_{i=\alpha_0+1}^s (\tau-W^i\wedge \tau)].
\end{equation}
Note that the mean number of servers inactive due to HOL blocking is: 
\begin{equation}
w_{HOL}=\lambda \mathbb{E}_0 \sum_{i=1}^{\alpha_0} (W^{\alpha_0}\wedge \tau-W^{i}\wedge \tau )\,.
\end{equation}

Another quantity of interest is the mean number of jobs in the system in steady state.
For this, we use Little's law which implies that this last mean number is equal to the mean value of $r$, defined above under the Palm distribution of the arrivals, multiplied by $\lambda$.

\section{Implementation of the algorithms}
\label{sec:implementation}

The implementations of Algorithm~\ref{alg:general} and Algorithm~\ref{alg:limit} are available in a public repository\footnote{https://github.com/UniVe-NeDS-Lab/MJSRE.git}. This section outlines the key aspects of their implementation and their performance characteristics.

\subsection{Implementation of Algorithm~\ref{alg:general}}

Algorithm~\ref{alg:general} returns an independent SPS of the stationary distribution of the system's workload. As discussed in Section~\ref{sec:metrics}, this is directly connected to the waiting times when the arrival process follows a Poisson law. 
SPSs can serve two main purposes: they can either be used as a starting point for a forward DES run, thereby eliminating the initial warm-up period, or numerous independent samples can be collected to directly estimate the distribution or moments of the waiting times. 

When choosing the sequence length $\ell$, there is a trade-off between computational efficiency and the estimate accuracy. Smaller values of $\ell$ allow more samples to be processed in parallel and reduce the number of doubling of the parameter operations performed by the algorithm. However, selecting an $\ell$ value that is too small increases the risk of stopping the execution too soon, potentially leading to less accurate estimates. This phenomenon is common in iterative algorithms. Empirically speaking, a practical guideline would be to choose $\ell$ large enough that even the rarest events occur a reasonable number of times within the sample computation. In this way, the impact of events with very low probability is taken into account, ensuring that estimates are representative of the system behavior.

\subsection{Parallelization and Performance Advantages}
A key feature of Algorithm~\ref{alg:general} is that two parallel executions perform identical operations on different instances of random variables, with the only difference being the halting condition, which may cause one execution to terminate earlier than the other. Our implementation takes advantage of this observation to parallelize the code, using the SIMD principle and run it in GPUs. 

The implementation is built using Python with the CuPy library, which enables CUDA-accelerated tensor computations with a NumPy-compatible interface. The workload vectors for a \emph{batch} of parallel simulation samples are stored in a single CuPy tensor. The corresponding sequences of random variables (service times, inter-arrival times) are similarly stored in tensors, having the number of considered batches as the first dimension and $\ell$ as the second one. The state update logic within the simulation loop is fully vectorized to eliminate any iterative processing over the samples and so that all considered samples are updated in a single computational step. Operations central to the algorithm, such as finding the maximum and sorting are executed as single, highly optimized CUDA kernels operating concurrently across the batch dimension. Even the convergence test is performed in a single step across all samples. This design ensures that the runtime is dominated by the loop on the sequence length given by $\ell$, rather than the number of samples, making the collection of thousands of samples approximately as fast as one. To manage memory constraints, the implementation allows one to reduce the number of batches considered simultaneously and to split sequences using the results discussed in Remark~\ref{rem:subseqeunces} with input parameters. Additionally, the code also takes advantage of the on-the-fly generation of random variable subsequences, ensuring that the degree of parallelization is limited only by the available hardware.

\subsection{Comparison Between MJSRE and Traditional DES}
The parallel approach of the MJSRE offers significant advantages with respect to traditional DES. The MJSRE iterates provide truly independent SPSs, allowing for the construction of unbiased and consistent estimators of the true SPS values with rigorous control of confidence intervals. In contrast, time series obtained with DES are typically temporally correlated, and one may run a limited number of independent experiments. This reduces the effective sample size and increases the variance of estimators. Moreover, the algebraic structure of the MJSRE matches very well the SIMD architecture of GPUs. This allows for very high degrees of parallelization, limited only by the capabilities of the available hardware. 
In contrast, since DES relies on a list of events, the nature of the $n$-th event (e.g., arrival, departure, etc.) depends on the samples of the random variables already sampled
(e.g., service times, inter-arrival times) needed to determine this event. This implies that DES can only be parallelized in a straightforward manner by distributing simulation runs across different CPUs, but not by using the SIMD/GPU type parallelism allowed by the SRE approach.

\subsubsection{Comparison of Execution Times}\label{subsec:exetimes}
Execution times, for both traditional DES and MJSRE, are highly dependent on several factors such as the implementation, the used hardware architecture, and system parameters themselves. As far as MJSRE is concerned, the main limiting factor in terms of performance remains the available GPU memory, which restricts the number of parallel samples that can be processed simultaneously. As a reference, we consider a 75\% system load with 20 servers and 5 job classes with service times {distributed according to a bounded Pareto}. Generating a single {SPS}  takes about one second, while a DES processing 5 million jobs completes in roughly 10 seconds. {In this context,} achieving comparable confidence intervals for the mean waiting time typically requires 30 DES runs (roughly 300 seconds) or 500{,}000 samples (roughly 10 seconds, with adequate hardware). Benchmarks were run on a workstation with an Intel\textsuperscript{\textregistered} i9-12900K CPU and an NVIDIA RTX 4080 GPU.
\subsection{Comparison Between Backward and Forward MJSRE Executions}
While the natural exploitation of the MJSRE is as in Algorithm~\ref{alg:general}, proceeding backward in time to generate lower bounds of perfect samples of the workload steady-state distribution, it is also possible to use the MJSRE for the simulation of the MJQM, proceeding forward in time. Indeed, by setting the initial workload vector $W$ to e.g. zero at time $t=0$, it is possible to run SIMD parallel simulations in the same manner as subsequences of workload vectors are computed in Algorithm~\ref{alg:general}, and to implement this on GPU as well. We call this type of simulation \emph{Forward MJSRE execution} (F-MJSRE), to distinguish it from Algorithm~\ref{alg:general}, which represents a \emph{Backward MJSRE execution} (B-MJSRE).
\\
In order to assess the effectiveness of the two approaches (B-MJSRE and F-MJSRE), we compare the results that are obtained in the study of a MJQM with N=256 servers, and two classes of jobs, one requiring just one server (small jobs), while the other requires all $N$ servers (big jobs). The job arrival process is Poisson with rate $\lambda$ jobs per second, and the probability of an arriving job to be big is $p_b=10^{-3}$. Job service times are exponentially distributed with rates $\mu_b=0.025$ services per second for big jobs and $\mu_s=1$ for small jobs.
In this simple context, we can compute the stability condition of the MJQM as $\lambda < 20.45 = \lambda_{\rm max}$, and, using the matrix-analytic approach described in~\cite{anggraito-mascots-2024}, we can numerically compute  the exact value of the average waiting time, that we consider as ground truth.
\\
In Table~\ref{tab:back-forw}, we report the average waiting time values computed by the B-MJSRE and F-MJSRE for job arrival rates $\lambda= k \times 0.1 \lambda_{\rm max}$, with $k=1,2,\cdots ,9$, together with their 99\% confidence intervals and their execution times, as well as the exact average values.
F-MJSRE execution times are almost constant, and correspond to the simulation of job arrival sequences that comprise on average 50 big jobs. B-MJSRE execution times are determined by the stopping condition of Algorithm~\ref{alg:general}, initialized with $\ell=10000$, and greatly depend on the system load.
\\
Results show quite interesting behaviors. First of all, we see 
that, for the considered setting, the B-MJSRE execution confidence intervals always include the exact values, while this is not true for the F-MJSRE at high load. This is paid by longer (much longer at high load) execution times, but while the B-MJSRE execution includes a stopping rule, this is not the case for the F-MJSRE execution. If we run the F-MJSRE for longer arrival sequences, the quality of the results increases, but we have no hint about the needed sequence length unless we first run the B-MJSRE (but if we do it, there is no point in running the F-MJSRE). For example, if we quadruple the number of simulated job arrivals, the highest considered system load still yields confidence intervals that do not include the exact value. If we run the F- and B-MJSRE with equal numbers of arrivals, we obtain similar results in terms of accuracy. We can also observe that the estimated average waiting times in the case of B-MJSRE always underestimate the exact value (as expected, since the B -MJSRE generates lower bounds of perfect samples), while this is not the case for F-MJSRE, that tends to overestimate at low loads and underestimate at high loads (the latter effect being likely due to the fact that the simulation has not yet reached stability).
Finally, the relative width of confidence intervals decreases for increasing load in both B-MJSRE and F-MJSRE. However, for arrival rates equal to 80\% and 90\% of the stability limit, this conveys a misleading message in the case of F-MJSRE. 
\\
In conclusion, the B-MJSRE execution uses the monotone-separable framework to provide statistical guarantees on the collected samples, while the F-MJSRE execution must rely on heuristics that can be sometimes misleading, especially when the system is analyzed in heavy load. 
\begin{table}[t]
    \caption{Comparison between the average waiting times estimates produced by the backward (BWD) and forward (FWD) executions of the MJSRE (all times in seconds)}
    \label{tab:back-forw}
    \centering
    \begin{tabular}{|c|c|c|c|c|c|c|c|}
    \hline
     & \multicolumn{3}{|c|}{Average waiting time} & \multicolumn{2}{|c|}{99\% confidence interval} & \multicolumn{2}{|c|}{Execution time} \\
    \hline
    $\lambda / \lambda_{\rm max}$     & Exact & BWD & FWD & BWD & FWD & BWD & FWD \\
    \hline
    0.1 & 3.75 & 3.40 & 3.93 & [2.99, 3.81] & [3.46, 4.40] & 410 & 698 \\
    \hline
    0.2 & 8.59 & 7.99 & 8.85 & [7.32, 8.65] & [8.11,  9.59] & 409 & 698 \\
    \hline
    0.3 & 14.93 & 14.07 & 15.37 & [13.14, 15.00] & [14.37, 16.37] & 408 & 696 \\
    \hline
    0.4 & 23.49 & 22.42 & 23.97 & [21.18, 23.66] & [22.67, 25.27] & 452 & 691 \\
    \hline
    0.5 & 35.59 & 34.49 & 36.03 & [32.87, 36.11] & [34.36, 37.71] & 1299 & 697 \\
    \hline
    0.6 & 53.86 & 52.77 & 54.27 & [50.60, 54.94] & [52.05, 56.49] & 3170 & 696 \\
    \hline
    0.7 & 84.47 & 82.13 & 84.18 & [80.09, 86.17] & [81.13, 87.22]  & 5920 & 701 \\
    \hline
    0.8 & 145.89 & 144.57 & 135.08 & [139.91, 149.24] & [130.81, 139.36] & 16094 & 694 \\
    \hline
    0.9 & 330.59 & 328.17 & 213.88 & [318.68, 337.65] & [208.10, 219.67] & 80801 & 695 \\
    \hline
    \end{tabular}
\end{table}

\subsection{Implementation of Algorithm~\ref{alg:limit}}

The implementation of Algorithm~\ref{alg:limit} is simpler than that of Algorithm~\ref{alg:general}. The stability condition is estimated by checking the convergence of the limit in Equation~(\ref{eq.2a}). The procedure used for its numerical computation employs a doubling scheme, similar to the one of Algorithm~\ref{alg:general}, which in this case continues until the difference between consecutive estimates satisfies a convergence threshold, i.e., becomes $\leq \epsilon$. Long sequences may cause numerical instability due to unbounded growth of vector $M$. We mitigate this by partitioning the sequence and incrementally normalizing $M$, storing offsets for accurate reconstruction. A mild CPU parallelization scheme is applied by distributing portions of $M$ between threads and using a parallel merge algorithm~\cite{cormen} for reordering. 

To evaluate execution times, Algorithm~\ref{alg:limit} was run on the same hardware setup with a 1\% relative accuracy target (i.e., $\varepsilon = \frac{0.01}{\lambda_{\text{ideal}}}$, where $\lambda_{\text{ideal}}$ is the stability region boundary with no HOL blocking, defined as $\lambda_{\text{ideal}} = s/\left(\sum_{i=1}^{C} \sigma_i\alpha_ip_i\right)$, where $C$ is the number of job classes and $p_i$ is the probability of seeing an arrival of class $i$). Under this setting and the same parameters of the scenario described in Section~\ref{subsec:exetimes}, the computation of the maximum arrival rate completed in about 22 seconds.

\section{Numerical results}\label{sec:numerical}
In this section, we present numerical results obtained with the MJSRE approach proposed in this paper, considering three MJQM configurations. We start by looking at the case of two job classes, for which results are present in the previous literature, so that we can validate our approach against existing results. We then consider a larger MJQM, with five job classes and 20 servers, exploring the effect of different job service time distributions. Finally, we study a realistic MJQM configuration, with 2048 servers and 17 job classes, whose characteristics are derived from the measurements on Google Borg data centers reported in~\cite{borg} and summarized in~\cite{baiocchi}. We consider job classes grouping jobs according to the number of tasks they are composed of. We assume that all tasks of a job are executed in parallel and that the duration of a job is is given by the duration of its longest task. 

In addition to presenting and discussing numerical results, this section serves the purpose of demonstrating the capabilities of the MJSRE approach in the performance analysis of MJQMs, allowing the computation of metrics that were impossible to obtain with the techniques available in the literature.

\subsection{MJQM with two job classes}
In Fig.~\ref{fig:pevacheck}, we plot the average number of servers wasted due to HOL blocking in a MJQM with $s=200$ servers, two equally likely job classes, one requiring just one server (small jobs), the other requiring 100 servers (big jobs). The results generated by the MJSRE are plotted together with the exact results obtained in~\cite{grosof-2020-arxiv} and~\cite{anggraito-mascots-2024} versus the average service time of big jobs. Since the results in the literature only concern the case of exponential service times and system saturation, we used the same assumptions for the MJSRE. We can see that the results obtained with the existing approaches and MJSRE match very well.
In addition to validating our method, Fig.~\ref{fig:pevacheck} clearly shows the complexity of this type of systems, where oscillating behaviors caused by the interplay of different factors make it extremely challenging to find a closed-form  expression for the stability condition, even in the simplest scenario of two classes with independent exponential job service time distributions.
\begin{figure}[!htbp]
  \centering
  \begin{subfigure}[t]{0.3\textwidth}
    \includegraphics[width=\linewidth]{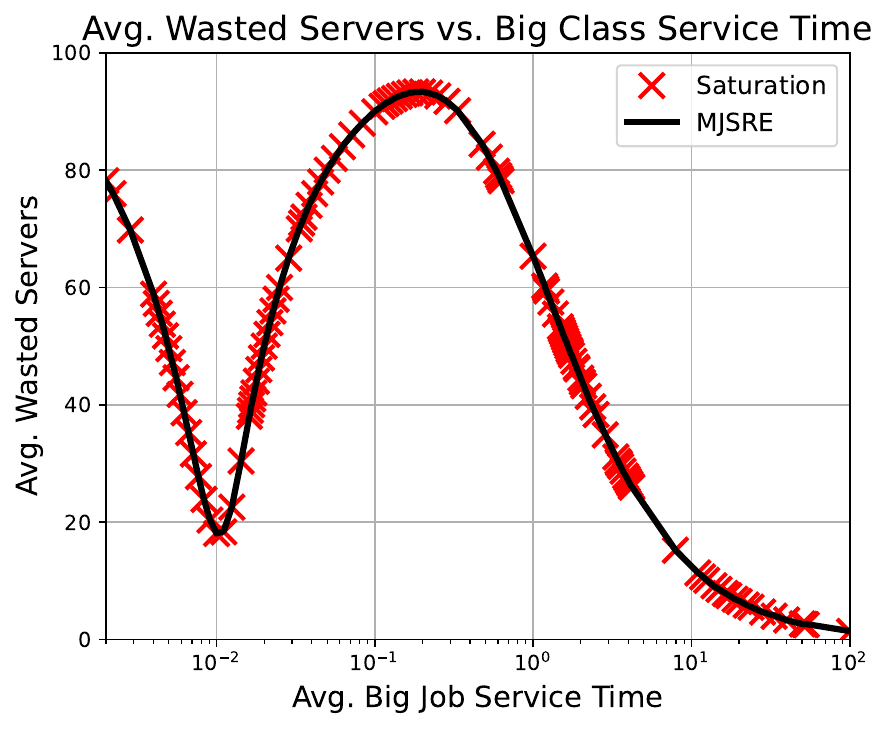}
    \caption{Average number of wasted servers vs. average service time of big jobs in a MJQM with two classes of jobs.} \label{fig:pevacheck}
  \end{subfigure}
  \hfill
  \begin{subfigure}[t]{0.3\textwidth}
     \includegraphics[width=\linewidth]{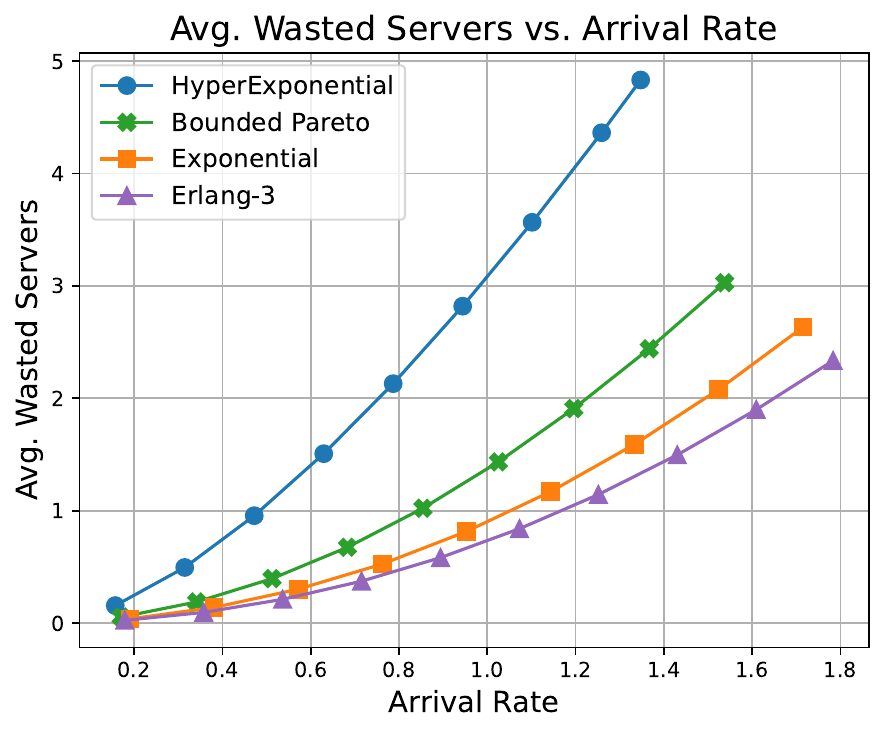}
    \caption{Average number of wasted servers vs. job arrival rate in a MJQM with five classes of jobs.}
    \label{fig:wasted_distro_comp}
  \end{subfigure}
  \hfill
  \begin{subfigure}[t]{0.3\textwidth}
    \includegraphics[width=\linewidth]{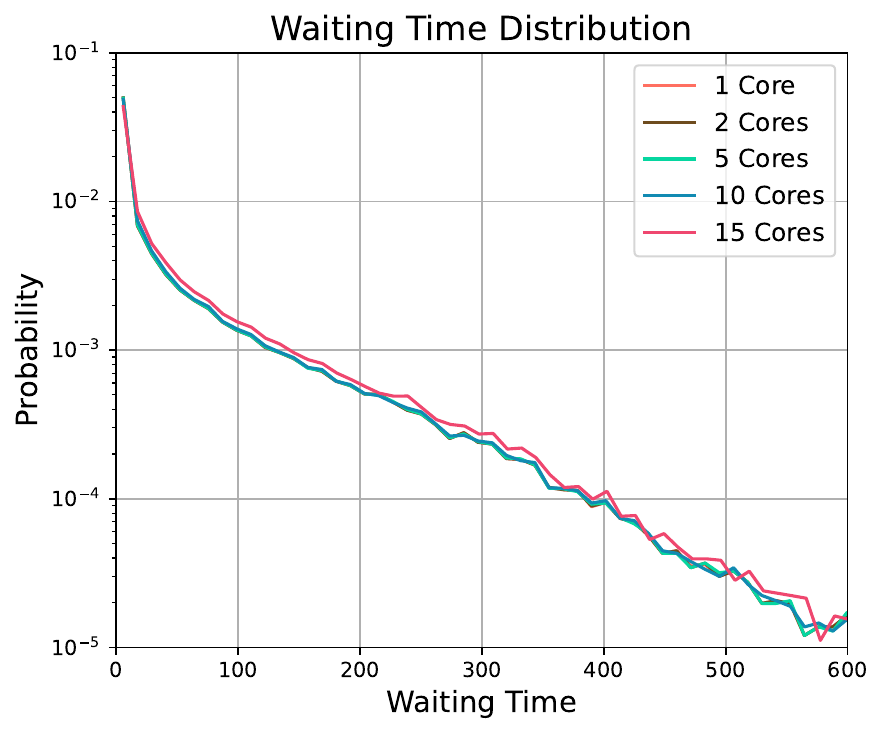}
    \caption{PDF of waiting times per class in a MJQM with five classes of jobs and bounded Pareto job service times.}
    \label{fig:classes_distros}
  \end{subfigure}
  \caption{Left and center: average number of servers wasted due to HOL, in a system with Poisson arrivals and exponential service times. Left: $s=200$, two job classes, comparison of literature  and MJSRE. Center: $s=20$, five job classes as in Table~\ref{tab:synth}. Right: Probability density function of waiting times per class in a MJQM with five classes of jobs and bounded Pareto job service times.
  \vspace{-0.2in}
  }
\label{fig:pevaandwasted}
  \Description{Exponential Service Times}
\end{figure}
\subsection{MJQM with five job classes}\label{subsec:synth}
Next, we present results for a MJQM configuration with 5 job classes. We consider a system equipped with $s=20$ identical servers. Job classes are defined by their number of core requirements, class probabilities, and mean job service times, as summarized in Table~\ref{tab:synth}. With these values, disregarding the effect of HOL, the system's stability condition would be 
$\lambda < \lambda_{\text{ideal}} = 2.11$.
\begin{table}[t]
\caption{Job class characteristics with corresponding parameters for hyperexponential (2-phase) and bounded Pareto service time distributions. Parameters for different distributions are chosen so as to maintain the same mean service time.}
{\footnotesize
\vspace{-0.1in}
\label{tab:synth}
\centering
\begin{tabular}{|l|c|c|c|c|c|c|}
\hline\specialcell[l]{\textbf{Cores}} & 
\textbf{Prob.}  & 
\specialcell[c]{\textbf{Mean}\\ \textbf{Service}\\ \textbf{Time}} & 
\specialcell[c]{\textbf{Mean} \textbf{Phase}\\\textbf{Time (Hyperexp)}} & 
\specialcell[c]{\textbf{Bounds}\\ \textbf{(B. Pareto)}} & 
\specialcell[c]{\textbf{Variance}\\\textbf{(Hyperexp)}} & 
\specialcell[c]{\textbf{Variance}\\\textbf{(B. Pareto)}}\\
\hline
1  & 0.5  & 0.5   & (25.0, 0.25)     & (0.15, 400.0) & 12.38 & 5.62 \\
2  & 0.1  & 0.83  & (41.5, 0.42)     & (0.25, 400.0) & 34.10 &  11.38 \\
5  & 0.2  & 1.25  & (62.5, 0.63)     & (0.38, 400.0) & 77.35 & 20.078 \\
10 & 0.19 & 3.33  & (166.5, 1.68)    & (1.05, 400.0)  & 548.96 & 77.15 \\
15 & 0.01 & 10.0  & (500.0, 5.05)    & (3.35, 400.0)  & 4950.51 & 335.56\\
\hline
\end{tabular}}
\end{table}
\begin{figure}[!htbp]
  \centering
  \begin{subfigure}[t]{0.3\textwidth}
    \includegraphics[width=\linewidth]{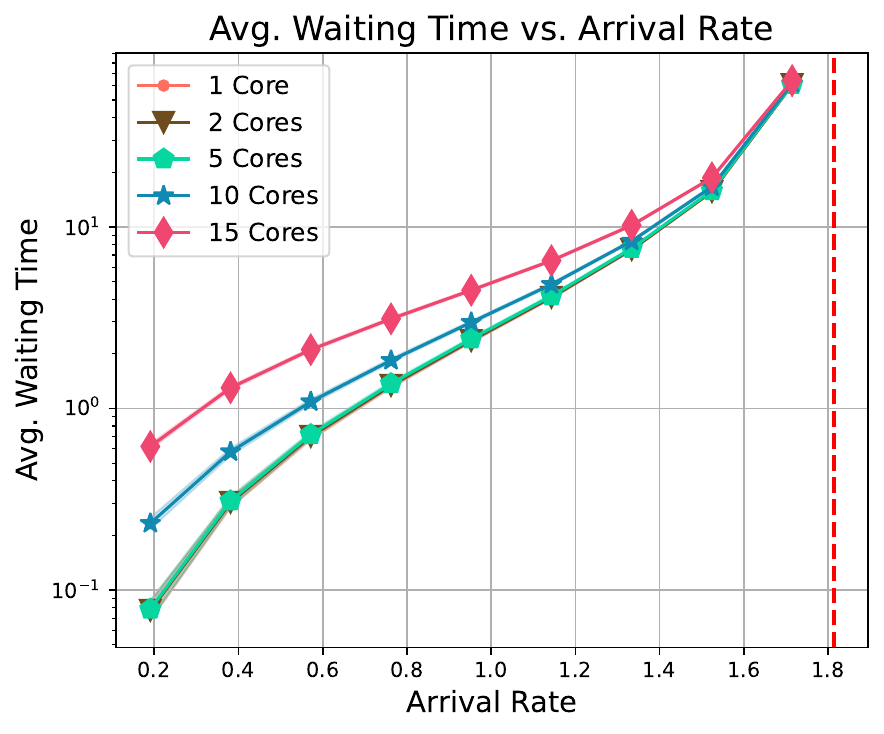}
    \caption{Average waiting time per class with 99\% confidence intervals.}
\label{fig:synth_avg_exp}
  \end{subfigure}
  \hfill
  \begin{subfigure}[t]{0.3\textwidth}
    \includegraphics[width=\linewidth]{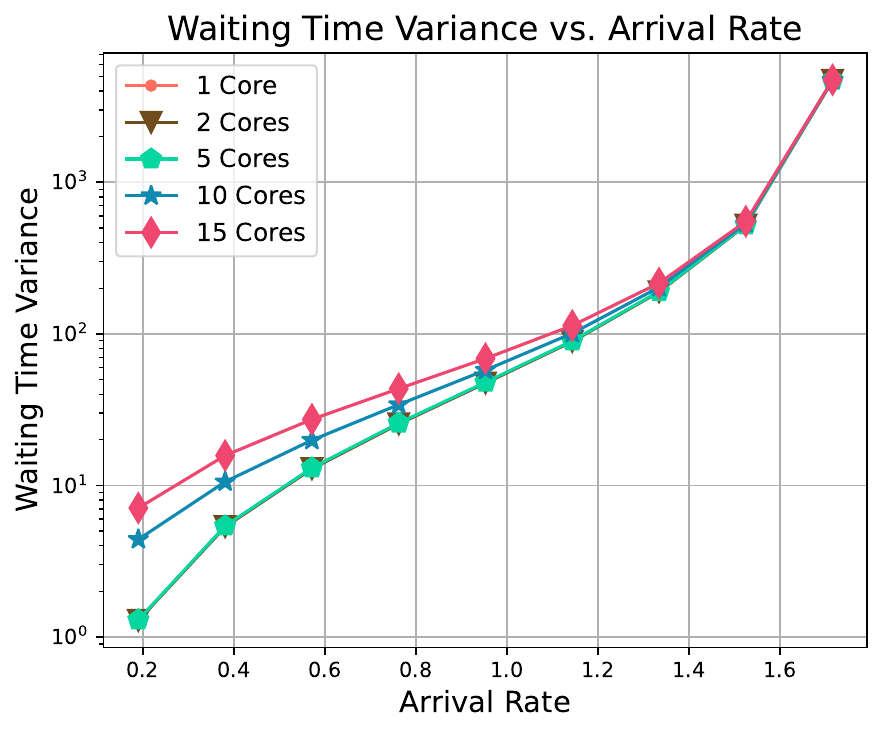}
    \caption{Variance of waiting times per class.}\label{fig:synth_var_exp}
  \end{subfigure}
  \hfill
  \begin{subfigure}[t]{0.3\textwidth}
    \includegraphics[width=\linewidth]{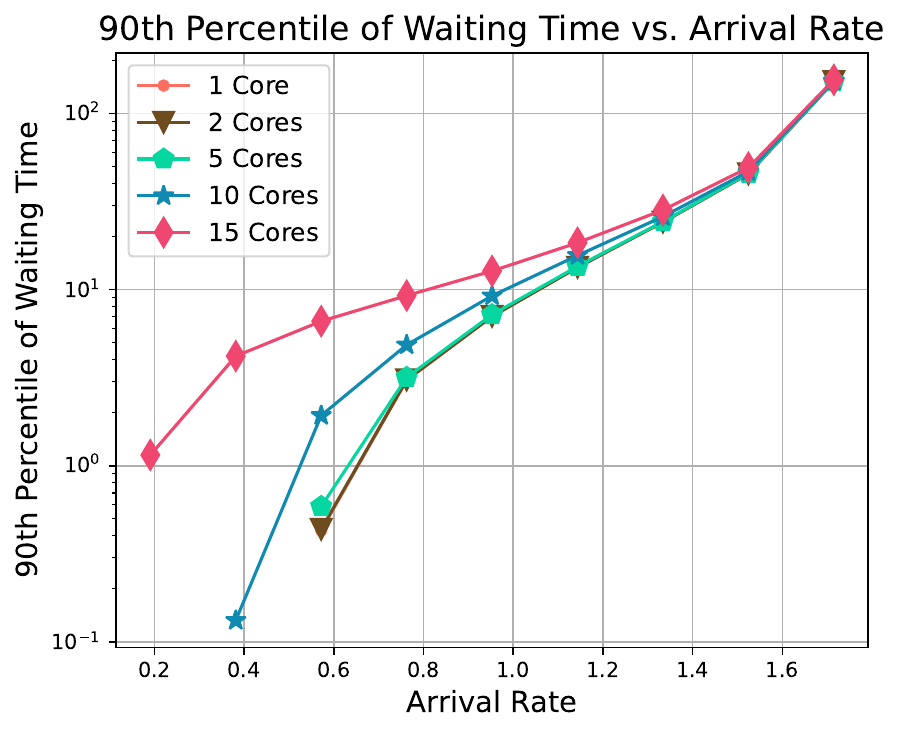}
    \caption{90th percentile of waiting times per class.}\label{fig:synth_90perc_exp}
  \end{subfigure}
  \vspace{-0.1in}
  \caption{Average, variance and 90-th percentile of per-class waiting times versus job arrival rate in the case of a MJQM with 20 servers and 5 job classes, with Poisson job arrivals and exponential job service times. The vertical bar in the left plot indicates the limit of the stability region.}\label{fig:synth_exp}
    \vspace{-0.1in}
  \Description{Exponential Service Times}
\end{figure}
We consider four different scenarios, always assuming Poisson job arrivals, but different job service time distributions:
\begin{enumerate}
    \item \textit{Exponential.} The service time of each job class follows an exponential distribution with rate equal to the inverse of the mean service time given in Table~\ref{tab:synth}. {The coefficient of variation of service times of jobs in each class is obviously 1, while the overall coefficient of variation of job service times is {$1.78$}. In this case the system is stable for {$\lambda<1.81$}. }
    \item \textit{Erlang (3-phases).} Job service times follow a 3-phase Erlang distribution with the same mean as the exponential. The coefficient of variation of jobs service times in each class is {$0.58$}, while the overall coefficient of variation of job service times is {$1.34$}. In this case, the system is stable for {$\lambda<1.84$}. 
    \item \textit{Hyperexponential (2-phases).} Job service times follow a 2-phase hyperexponential distribution. For each job class, phase rates are set so as to match the mean job service time of the exponential case. The resulting mean phase times are reported in Table~\ref{tab:synth}, where the longer phase is chosen with probability $0.01$.
    The coefficient of variation of service times of jobs in each class is {$7.036$}, while the overall coefficient of variation of job service times is {$10.22$}. In this case, the system is stable for {$\lambda<1.50$}. 
    \item \textit{Bounded Pareto.} Job service times follow a bounded Pareto distribution with lower and upper bounds $(x_{\text{min}}, x_{\text{max}})$ as described in Table~\ref{tab:synth} and shape parameter set to $1.4$ for all classes. The coefficients of variation of jobs service times in each class are $4.74, 4.06, 3.58, 2.64, 1.83$ when ordering them in increasing core requirement, while the overall coefficient of variation of job service times is {$4.01$}. In this case, the system is stable for {$\lambda<1.72$}. 
    \end{enumerate}
For this MJQM with 5 job classes, in Fig.~\ref{fig:wasted_distro_comp},
we plot the average number of servers wasted due to HOL blocking versus the job arrival rate, with the different distributions of job service times. We can observe that the average number of wasted servers increases with the system load, and with the variance of service times. With hyperexponentially distributed service times, the average number of wasted servers gets close to 5, i.e., to 25\% of the overall MJQM service capacity. In Fig.~\ref{fig:classes_distros}, we plot the probability density function of waiting times per job class with bounded Pareto job service times. The curves show that differences among classes are small and their linear decay indicates the presence of a heavy tail.

Figure~\ref{fig:synth_exp} shows the average, variance, and 90th percentile of job waiting times for each of the 5 classes, plotted against the job arrival rate. The red dashed vertical line reports the stability limit for the MJQM.
We can see that performance metrics are very similar for the 3 classes with lower core requirements (1, 2 and 5 cores), and somewhat higher for jobs of the classes requiring 10 and 15 servers.
The same type of results is plotted in Fig.~\ref{fig:synth_hyperexp} for hyperexponential distributions of job service times. In this case of higher service time variance, the waiting time metrics are at least one order of magnitude higher.
In order to better visualize the impact of the job service time distributions, Figures~\ref{fig:synth_T1} and~\ref{fig:synth_T15} show the average, variance, and 90th percentile of waiting times, as functions of the job arrival rate, for the job classes requiring either 1 or 15 servers. These plots also confirm that performance metrics increase with both system load and variance of job service times.

\subsection{The Google Borg Cell B Scenario}\label{subsec:cellB} 

For a final demonstration of the capabilities of the MJSRE approach, we use a processed version of the data presented in~\cite{borg, baiocchi}. We consider a MJQM with $s=2048$ and 17 job classes, whose characteristics are reported in Table~\ref{tab:cellBreduced}.
These values correspond to those measured on one cell (Cell B) of a Google Borg data center~\cite{borg}.
In Fig.~\ref{fig:cellBreduced}, we plot the curves of the overall average job waiting time versus the job arrival rate.
The red dashed vertical line reports the stability limit for the MJQM.
The plot shows that the average waiting time results obtained with the MJSRE are quite close to the predictions of traditional DES.
\begin{figure}[!htbp]
  \centering
  \begin{subfigure}[t]{0.3\textwidth}
    \includegraphics[width=\linewidth]{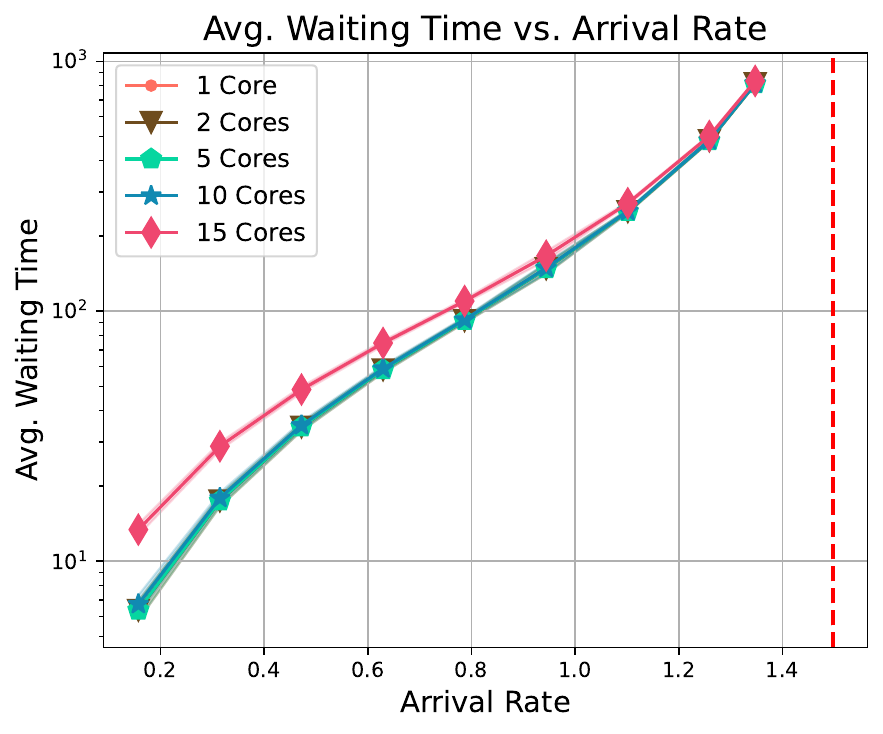}
    \caption{Average waiting time per class with 99\% confidence intervals.}\label{fig:synth_avg_hyperexp}
  \end{subfigure}
  \hfill
  \begin{subfigure}[t]{0.3\textwidth}
    \includegraphics[width=\linewidth]{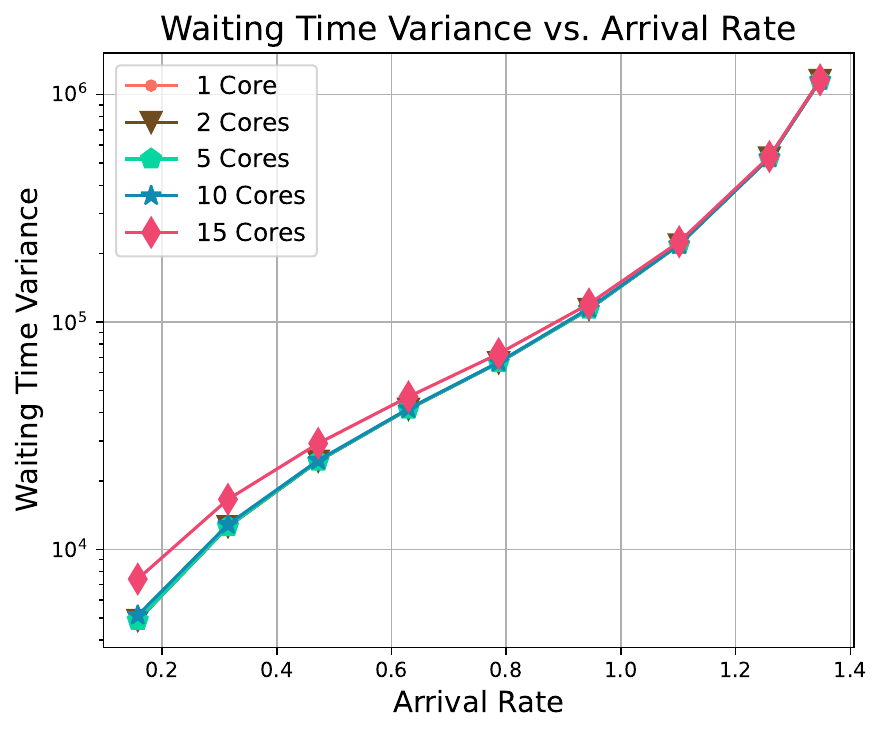}
    \caption{Variance of waiting times per class.}\label{fig:synth_var_hyperexp}
  \end{subfigure}
  \hfill
  \begin{subfigure}[t]{0.3\textwidth}
    \includegraphics[width=\linewidth]{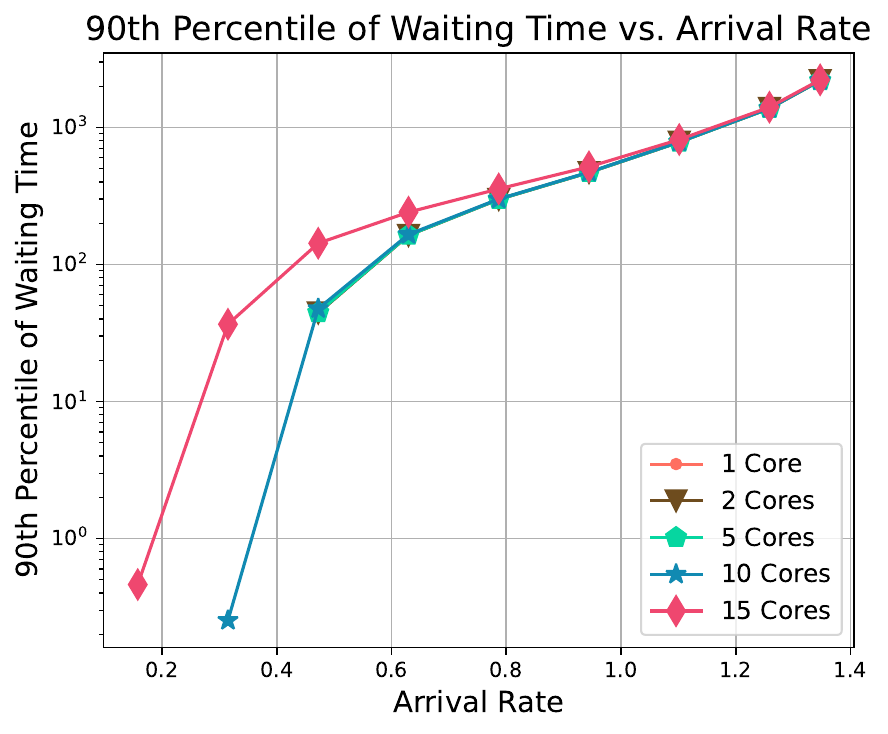}
    \caption{90th percentile of waiting times per class.}\label{fig:synth_90perc_hyperexp}
  \end{subfigure}
  \vspace{-0.2in}
  \caption{Same as Fig.~\ref{fig:synth_exp} with hyperexponential service times.
    \vspace{-0.2in}
}\label{fig:synth_hyperexp}
  \Description{HyperExponential Service Times}
\end{figure}
\begin{figure}[!htbp]
  \centering
  \begin{subfigure}[t]{0.3\textwidth}
    \includegraphics[width=\linewidth]{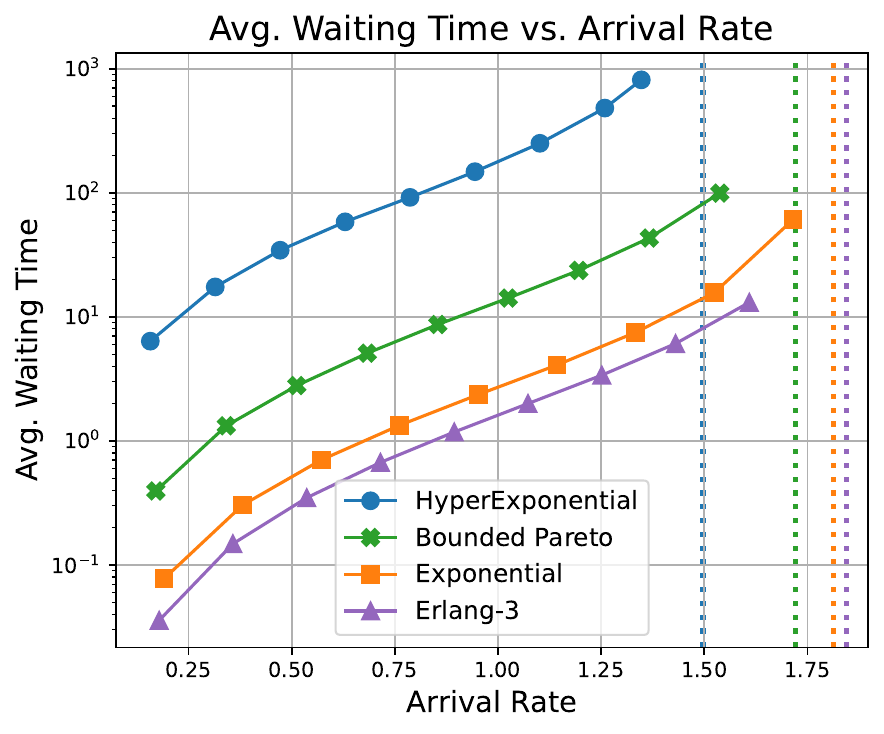}
    \caption{Average Waiting Time vs. Arrival Rate.}\label{fig:synth_avg_T1}
  \end{subfigure}
  \hfill
  \begin{subfigure}[t]{0.3\textwidth}
    \includegraphics[width=\linewidth]{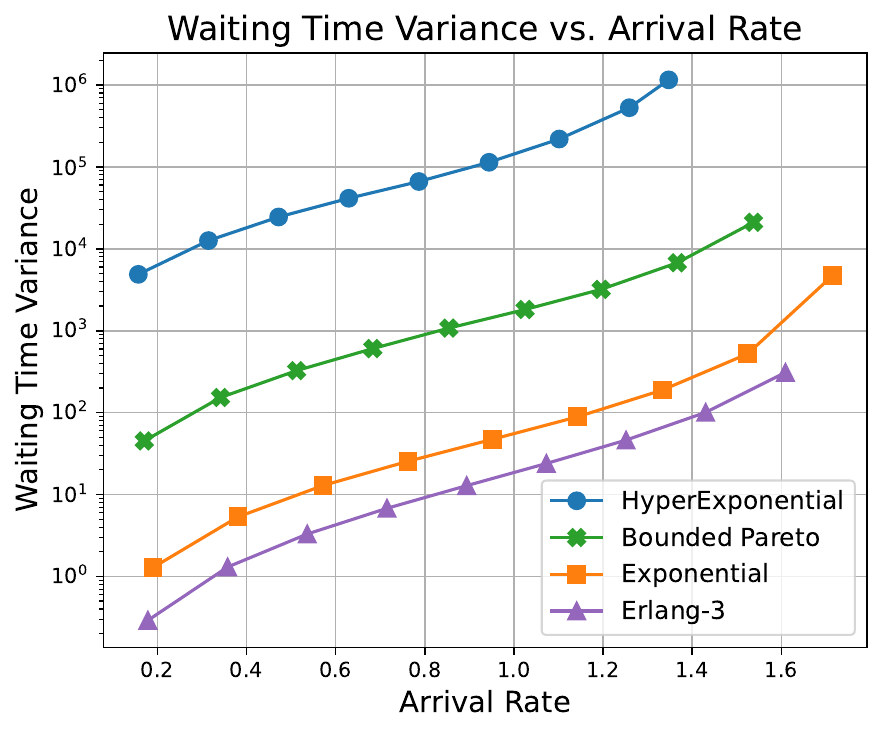}
    \caption{Waiting Times Variance vs. Arrival Rate.}\label{fig:synth_var_T1}
  \end{subfigure}
  \hfill
  \begin{subfigure}[t]{0.3\textwidth}
    \includegraphics[width=\linewidth]{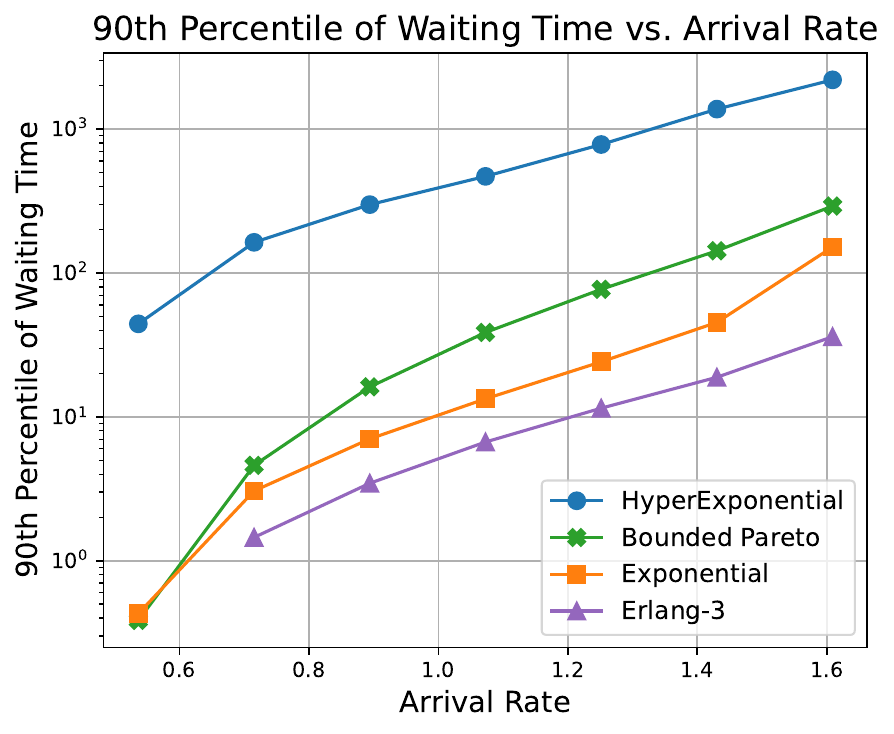}
    \caption{90th percentile of Waiting Times vs. Arrival Rate.}\label{fig:synth_90perc_T1}
  \end{subfigure}
    \vspace{-0.1in}
  \caption{Average, variance and 90-th percentile of waiting times versus job arrival rate for jobs requiring just 1 server in the case of a MJQM with 20 servers and 5 job classes. This is for Poisson job arrivals and job service times with exponential, hyperexponential, Erlang-3 or bounded Pareto distribution, as in Table~\ref{tab:synth}. The vertical bars in the left plot indicate the limit of the stability region with the different service time distributions.}\label{fig:synth_T1}
  \Description{Average, variance and 90-th percentile of waiting times versus job arrival rate for jobs requiring just 1 server in the case of a MJQM with 20 servers and 5 job classes. Poisson job arrivals and job service times with exponential, hyperexponential, Erlang-3 or bounded Pareto distribution, as in Table~\ref{tab:synth}. The vertical bars in the left plot indicate the limit of the stability region with the different service time distributions.}
\end{figure}
\begin{figure}[!htbp]
  \centering
  \begin{subfigure}[t]{0.3\textwidth}
    \includegraphics[width=\linewidth]{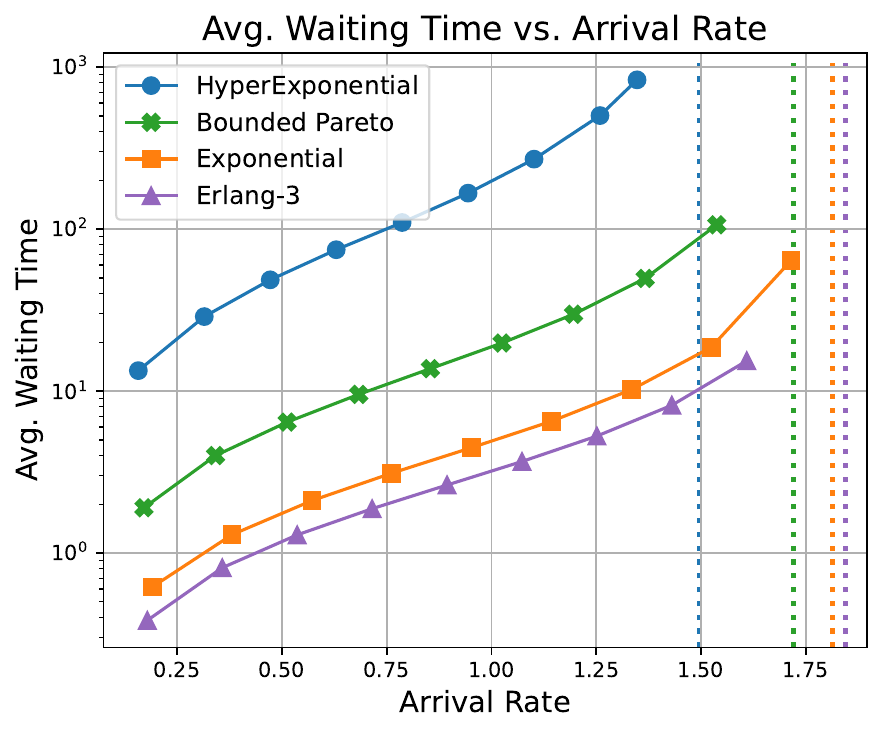}
    \caption{Average Waiting Time vs. Arrival Rate.}\label{fig:synth_avg_T15}
  \end{subfigure}
  \hfill
  \begin{subfigure}[t]{0.3\textwidth}
    \includegraphics[width=\linewidth]{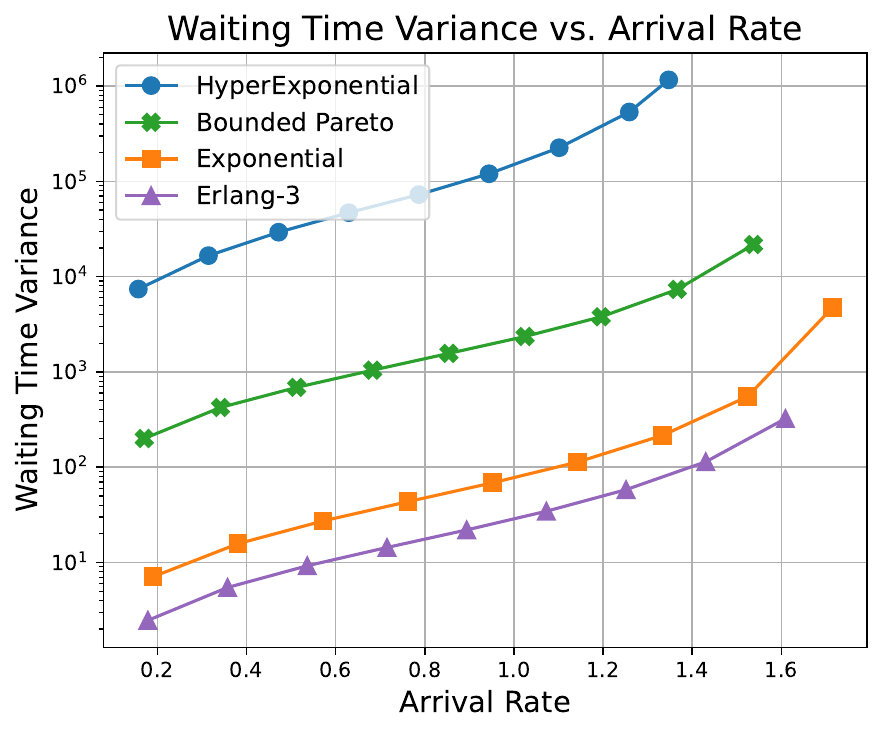}
    \caption{Waiting Times Variance vs. Arrival Rate.}\label{fig:synth_var_T15}
  \end{subfigure}
  \hfill
  \begin{subfigure}[t]{0.3\textwidth}
    \includegraphics[width=\linewidth]{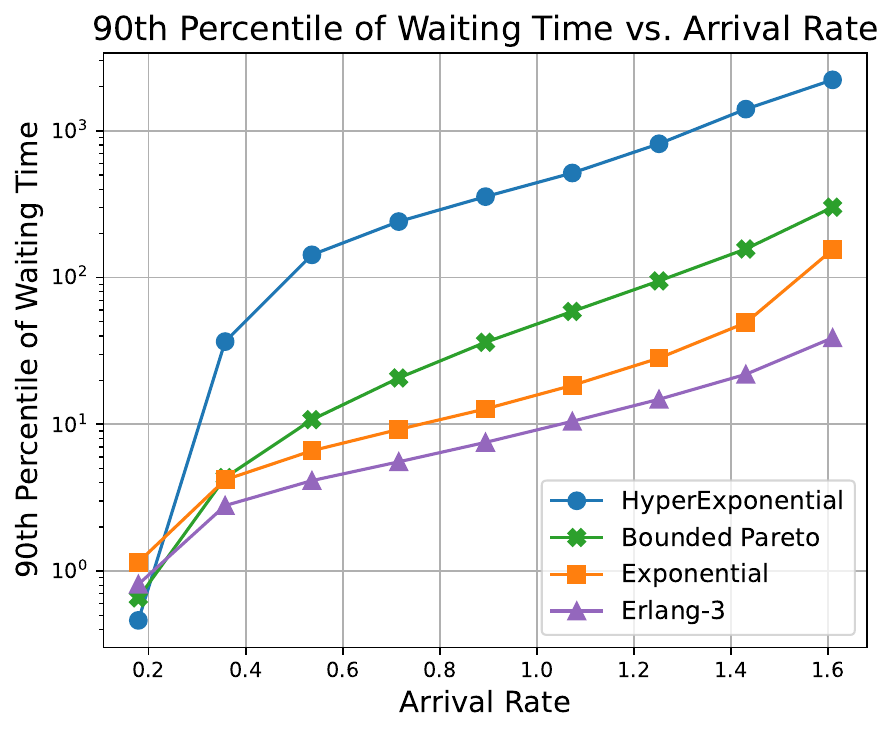}
    \caption{90th percentile of Waiting Times vs. Arrival Rate.}\label{fig:synth_90perc_T15}
  \end{subfigure}
  \vspace{-0.1in}
  \caption{Same as Fig.~\ref{fig:synth_T1} for jobs in the class requiring 15 servers.}
    \vspace{-0.1in}
\label{fig:synth_T15}
  \Description{Class T15}
\end{figure}
\begin{figure}[!htbp]
    \centering
    \begin{subfigure}[t]{0.55\textwidth}
        \centering
        \includegraphics[width=\linewidth]{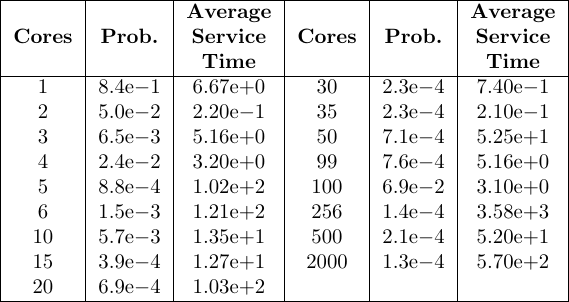}
        \caption{Cell B workload parameters.}
        \label{tab:cellBreduced}
    \end{subfigure}%
    \hfill
    \begin{subfigure}[t]{0.43\textwidth}
        \centering
        \includegraphics[width=0.85\linewidth]{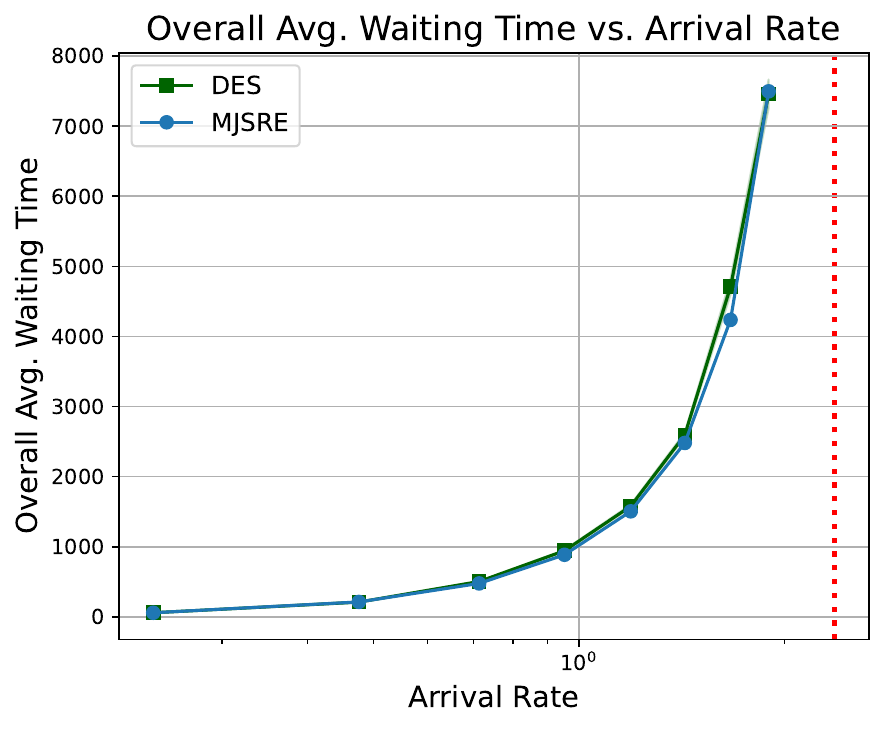}
        \caption{Average waiting time vs. job arrival rate in the case of Cell B.}
        \label{fig:cellBreduced}
    \end{subfigure}
  \vspace{-0.1in}
    \caption{Google Borg Cell B: parameters in the left table and average overall waiting times vs. job arrival rate in the plot on the right.}
    \Description{Google Borg Cell B: parameters in the left table and average waiting times vs. job arrival rate in the plot on the right}
\end{figure}
\section{Generalizations of the Multiserver-Job Recurrence Equation}\label{sec:generalizations}
In this section, we show that the framework proposed for MJQM with FCFS discipline has a wide area of applicability. We propose two extensions of practical interest for which monotonicity holds, hence the discussion proposed so far for the production of SPSs can be immediately extended.
\subsection{Requests for specific resources}
In many practical scenarios, jobs are not oblivious to the resources that they obtain. For example, they may ask for a specific server in a farm because it is equipped with a special GPU or because it has a considerable amount of memory that satisfies the job computational needs. Hereafter, we consider this scenario where servers have identities and jobs may require either an arbitrary server or a specific one. The service discipline is always FCFS. More formally, 
assume now that job $n$ requires $\alpha_n$ non-specified CPUs and in addition $\beta_n$
specific CPUs with $\alpha_n+\beta_n\le s$.
Here, we will assume that the required specific CPUs are picked uniformly
at random within the set $[1,\ldots,s]$, forming a random subset $S_n$ of the latter of cardinality $\beta_n$ of $[1,\dots, s]$.
Then the ordered workload vector satisfies the Specific Multiserver-Job Stochastic Recurrence Equation (SMJSRE):
\begin{equation}\label{eq.3}
        W_{n+1} = {\mathcal R}( V(\alpha_n, S_n, W_n) + \sigma_n  F(\alpha_n, S_n) - \tau_n  \un)^+\,.
\end{equation}
The vector $F(\alpha,S)\in \mathbb R^s$ is defined as follows.
Let  
$
k(S,i)=\sum_{j=1}^i 1_{j\notin S}$  for all  $1\le i\le s$.

Then:
$$
\begin{cases}
F(\alpha,S)^i = & 1,\quad \mbox{for all $i\in S$ or all $i\notin S$ s.t. $k(S,i)\le \alpha$};\\
	& 0,\quad \mbox{otherwise}\,.
\end{cases}
$$
The mapping $V(\alpha,S,W)$ takes as input an 
ordered vector $W\in \mathbb R^s$ and returns the vector of $\mathbb R ^s:$ 
$$
\begin{cases}
	V(\alpha,S,W)^i= & \hat W ,\quad \mbox{for all $i\in S$ or all $i\notin S$ s.t. $k(S,i)\le \alpha$};\\
        & W^i ,\quad \mbox{otherwise}\,,
\end{cases}
$$
with $\hat W= \max(\max_{j\in S} W^j,W^{i^*})$
and $i^*=\inf\{i\ s.t.\ k(i)=\alpha\}$.

\subsection{Multiple resource types}

Real data centers have multiple types of resources (memory, cores, GPUs, etc..) that are available in a finite quantity. Suppose that we have $T$ types of resources. Then, their availability is represented by a vector of positive integers $\overline s=(s^1, \ldots, s^T)$. Requests are described by their inter-arrival times $\tau_n \in \mathbb R^{>0}$, their service times $\sigma_n \in \mathbb R^{>0}$ and their resource demand $\overline \alpha_n = (\alpha_{n}^1, \ldots, \alpha_{n}^T)$, with $0\leq \alpha_n^i \leq s_i$ for all $i=1, \ldots, T$. The workload at the $n$-th job arrival, denoted by $\overline W_n$, is a collection of $T$ ordered vectors of non-negative real entries, i.e.,
$
\overline W_n=(W_{n,1}, \ldots, W_{n,T})\,,
$
where $W_{n,i}$ has $s_i$ components. We say that $\overline W_n \leq \overline W_n'$ if, for all $i=1, \ldots, T$, we have that $W_{n,i}\leq W_{n,i}'$, where the last inequality, as in Section~\ref{sec:mjsre}, must be considered componentwise. Henceforth, we omit the subscript $n$ when we consider a generic workload $\overline W$. 

The SRE for this case is similar to Equation~(\ref{eq.2}), but it now accounts for multiple resource types and will be referred to as the 
Multiresource Multiserver-Job Stochastic Recurrence Equation (MMJSRE):
\begin{equation}\label{eq:msre.mr}
	\overline W_{n+1} = {\mathcal R}( \overline \sync(\overline \alpha_n, \overline W_n) + \sigma_n  \overline \load(\overline \alpha_n) - \tau_n  \un)^+,
\end{equation}
where $\overline \load(\overline \alpha)=(L(\alpha_1), \ldots, L(\alpha_T))$, and $L(\alpha_i)$ is defined in Section~\ref{sec:mjsre} and has $s_i$ components. 
For all ordered $\overline W$, $\overline \sync(\alpha, W)\in (\mathbb{R}^{s_1} \times \cdots \times \mathbb R^{s_T})$ is defined by
$
\overline \sync (\overline \alpha, \overline W)=(\sync_1(\overline \alpha, \overline W), \ldots,
\sync_T(\overline \alpha, \overline W))\,,
$
and:
$$
\sync_j(\overline \alpha, \overline W)^i = 
\begin{cases}
W_1^{\alpha^1}\vee \cdots \vee W_T^{\alpha^T}, & \mbox{for all $1\le i\le \alpha_j$}\\
W_j^i & \mbox{for all $i=\alpha_j+1,\ldots, s_j$\;.}
\end{cases}
$$
$U$ is defined coherently with the dimensions of $\overline W$ and the ascending reordering operator $\mathcal R$ is also applied to each vector $W_{n,i}$ individually.

The first essential property that we need to develop to apply the theory proposed in this paper is the analogue of Lemma~\ref{lemma:monotonic} for the SRE~(\ref{eq:msre.mr}). Now, suppose we have two workload vectors $\overline W \leq \overline W'$. Equation~(\ref{eq:msre.mr}) modifies the first $\alpha_i$ components of entries $W_i$ and $W_i'$. By hypothesis $W_1^{\alpha^1}\vee \cdots \vee W_T^{\alpha^T} \leq W_1^{'\alpha^1}\vee \cdots \vee W_T^{'\alpha^T}$. Then, the proof follows as for Lemma~\ref{lemma:monotonic} since the reordering is applied independently to each component of $\overline W$ and $\overline W'$.   

\subsection{Special cases and relations with other stochastic recurrence equations}
The MMJSRE contains the MJSRE as a special case, which in turn contains the G/G/s queue SRE as a special case.
The MJSRE also contains the Task Graph recurrence equations of~\cite{QMIPS-book, QMIPS-article} which is an instance
of stochastic (max,plus) linear recurrence equation~\cite{BCOQ}.

\section{Conclusion}\label{sec:conclusion}
The performance analysis of MJQM has drawn significant attention from the research community in recent years. However, existing analytical and numerical results are limited to cases where the number of classes is unrealistically low (most often just $2$) and where Markovian assumptions are used. 
While these results provide insights into the issue of low maximum achievable throughput in these systems, they do not capture the realistic data center dynamics where numerous classes of jobs are active and the processing or inter-arrival times are not exponential. 

DES is also challenging because some job classes may have extremely low arrival probability (thus requiring very long simulations). Even worse, the stability condition is not known for more than two classes even under the exponential assumption on the service times and independent Poisson arrival process. As a result, before running a simulation experiment, it is not possible to know whether the system is stable under a given traffic intensity, and hence can reach steady state. 

In this work, we have proposed a novel approach to MJQM analysis with FCFS scheduling that allowed us to provide solutions to the problems mentioned above. We have first reduced this class of dynamics to stochastic recurrence equations. We have then proved two key properties of the stochastic recurrence equation underlying these systems: a) the monotonicity of the ordered workload at the servers and b) the fact that the system satisfies the properties of the monotone-separable framework. As a consequence of property a) we have developed a massively parallelizable algorithm to draw SPSs of the system's workload. In turns, this allowed us to compute important performance metrics of the MJQM, such as the distribution of the waiting times per class, the average number of wasted servers, the expected number of jobs in the system in steady state. 
The fact that this system is monotone-separable has further important implications, in particular the fact that the stability analysis can be carried out by decoupling the arrival process from the service process. This allowed us to obtain the first (to the best of our knowledge) characterization of the conditions under which the waiting time in MJQMs has a finite stationary regime, as well as a finite second moment, thus allowing confidence intervals to be built around its mean. Second, we have given the first algorithm for the numerical computation of the stability condition of FCFS MJQMs with arbitrary number of classes, arbitrary service time distribution (with finite mean) and arbitrary inter-arrival time distribution. Quite interestingly, the stability condition depends on service time distributions but is insensitive to the inter-arrival time distribution {as long as the probabilities of job classes are independent of the MJQM state}. Intuitively, this happens because service time distributions 
influence how jobs release resources, and hence a higher variability 
determines a lower maximum utilization, while inter-arrival time distributions do not affect maximum utilization. 

This paper opens several lines of research. Although we have shown in this work that this framework can be extended to systems with multiple resource types or with requests to specific resources, an open problem remains: \textit{which scheduling disciplines beyond FCFS satisfy the monotonicity property?} From the point of view of algorithm implementation, a relevant future direction of research concerns further parallelization of the MJSRE algorithms. Another promising line of research concerns the replacement of SPS
by perfect samples along the lines of what is available for the GI/GI/c queue~\cite{blanchet}. The question whether there exist
computable stationary upper-bounds that have the same stability condition as the initial MJQM is open.
It is also particularly appealing to try to exploit the potential parallelism in the number of servers $s$ when $s$ is large. The algebraic structure identified in this paper opens a new path in this direction along the lines of what was done, for example, in~\cite{Canales}.

\section{Acknowledgments}
This work was supported in part by Vinnova Center for Trustworthy Edge Computing Systems and Applications (TECoSA), in part by ERC NEMO under the European Union’s Horizon 2020 research and innovation program (grant agreement number 788851 to INRIA), and in part by Project TUCAN6-CM under Grant TEC-2024/COM-460 funded by the Community of Madrid (ORDER 5696/2024). M. Ajmone Marsan and F. Baccelli thank LINCS where this research was initiated.


\bibliographystyle{ACM-Reference-Format}
\bibliography{biblio}

@article{afanaseva-2020,
    author = {Afanaseva, L. and Bashtova, E. and Grishunina, S.}, 
    title = {{Stability Analysis of a Multi-server Model with Simultaneous Service and a Regenerative Input Flow}},
    year = {2020},
    volume = {22}, 
    journal = {Methodology and Computing in Applied Probability}, 
    pages = {1439--1455}
}

@inproceedings{anggraito-mascots-2024,
    author={Anggraito, A. and Olliaro, D. and Marin, A. and {Ajmone Marsan}, M.},
    booktitle={2024 MASCOTS}, 
    title={{The Non-Saturated Multiserver Job Queuing Model with Two Job Classes: a Matrix Geometric Analysis}}, 
    year={2024},
    publisher = {IEEE},
    address = {New York, USA},
    pages={1--8}
}

@article{anggraito-cox-2025,
	title = {The Multiserver Job Queuing Model with two job classes and Cox-2 service times},
	journal = {Performance Evaluation},
	volume = {169},
	pages = {102486},
	year = {2025},
	author = {Anggraito, A. and Olliaro, D. and Marin, A. and {Ajmone Marsan}, M.},
}

@book{BB,
    author = {Baccelli, F. and Brémaud, P.},
    year = {1994},
    title = {Elements of Queueing Theory},
    publisher = {Springer},
    address = {Heidelberg, Germany}
}

@article{Canales,
    author = {Baccelli, F. and Canales, M.},
    title = {Parallel Simulation of Stochastic {Petri} Nets Using Recursive Equations}, 
    journal = {ACM Transactions on Modeling and Computer Simulation (TOMACS)},
    volume ={3(1)}, 
    pages = {20-41}, 
    year = {1993}
}

@book{BCOQ,
    author = {Baccelli, F. and Cohen, G. and Olsder, G.J. and Quadrat, J-P.},
    year = {1992},
    title = {Synchronization and Linearity},
    publisher = {Wiley},
    address = {Hoboken, NJ, USA}
}

@article{BCM,
    title = {Construction of the stationary regime of queues with locking},
    journal = {Stochastic Processes and their Applications},
    volume = {26},
    pages = {257-265},
    year = {1987},
    author = {Baccelli, F. and Courcoubetis, C. A. and Reiman, M. I.},
}

@article{BF,
    author = {Baccelli, F. and Foss, S.},
    journal = {Journal of Applied Probability},
    pages = {494--507},
    title = {On the Saturation Rule for the Stability of Queues},
    volume = {32},
    year = {1995}
}

@article{BF2,
    author = {Baccelli, F. and Foss, S.},
    journal = {Annals of Applied Probability},
    title ={Moments and Tails in Monotone-Separable Stochastic Networks},
    pages = {621-650},
    volume = {14},
    year = {2004}
}

@Inbook{QMIPS-article,
    author = {Baccelli, F. and Gaujal, B. and Jean-Marie, A. and Mairesse, J.},
    title = {Analysis of Parallel Processing Systems via the (max,+) Algebra},
    bookTitle = {Quantitative Methods in Parallel Systems},
    year = {1995},
    publisher = {Springer},
    address = {Heidelberg, Germany},
    pages = {69--98},
}

@book{QMIPS-book,
    author = {Baccelli, F. and Jean-Marie, A. and Mitrani, I.},
    year = {1995},
    title = {Quantitative Methods in Parallel Systems},
    publisher = {Springer},
    address = {Heidelberg, Germany}
}

@article{BL,
    author = {F. Baccelli and Z. Liu},
    journal = {Journal of the ACM},
    title ={On the Execution of Parallel Programs on Multiprocessor Systems - A Queueing Theory Approach},
    pages = {373--414},
    volume = {32},
    year = {1990}
}

@inproceedings{beloglazov,
    author = {Beloglazov, A. and Buyya, R.},
    booktitle = {2010 10th IEEE/ACM International Conference on Cluster, Cloud and Grid Computing},
    publisher = {IEEE},
    pages = {577--578},
    title = {{Energy Efficient Allocation of Virtual Machines in Cloud Data Centers}},
    address = {New York, USA},
    year = {2010}
}

@article{blanchet,
    author = {Blanchet, J. and Pei, Y. and Sigman, K.},
    title = {Exact sampling for some multi-dimensional queueing models with renewal input},
    journal = {Advances in Applied Probability},
    pages = {1179--1208},
    year = {2019},
    volume = {51}
}

@article{blanchet0,
    title = {Exact sampling of stationary and time-reversed queues},
    author = {Blanchet, J. and Wallwater, A.},
    journal = {ACM TOMACS}, 
    volume = {25},  
    numpages = {27}, 
    year  = {2015}
}

@article{Bramson,
    author = {Bramson, M.},
    title = {{Stability of queueing networks}},
    volume = {5},
    journal = {Probability Surveys},
    pages = {169--345},
    year = {2008}
}

@article{brill-1984,
    author = {Brill, P.H. and Green, L.},
    journal = {Management Science},
    pages = {51--68},
    title = {{Queues in which Customers receive Simultaneous Service from a Random Number of Servers: a System Point Approach}},
    volume = {30},
    year = {1984}
}

@book{cormen,
    author = {Cormen, T. H. and Leiserson, C. E. and Rivest, R. L. and Stein, C.},
    year = {2009},
    title = {Introduction to Algorithms (3rd ed.)},
    publisher = {MIT Press and McGraw-Hill},
    address = {New York, USA}
}

@article{karatza,
    author  ={Filippopoulos, D. and Karatza, H.},
    year = {2007},
    title = {{An {M/M/2} parallel system model with pure space sharing among rigid jobs}},
    journal  ={Mathematical Computer Modelling},
    volume = {45},  
    pages = {491--530}
}

@inproceedings{grosof-2023-serverfilling,
    author = {Grosof, I. and Harchol-Balter, M.},
    booktitle = {{Proceedings of the 5th Workshop on Advanced Tools, Programming Languages, and PLatforms for Implementing and Evaluating Algorithms for Distributed Systems }},
    publisher = {Association for Computing Machinery},
    pages = {1--5},
    title = {{Invited Paper: ServerFilling: A Better Approach to Packing Multiserver Jobs}},
    address = {New York, USA},
    year = {2023}
}

@misc{grosof-2020-arxiv,
      title = {{Stability for Two-class Multiserver-job Systems}}, 
      author = {Grosof, I. and Harchol-Balter, M. and Scheller-Wolf, A.},
      note = {\emph{arXiv preprint:} 2010.00631},
      year = {2020}
}

@article{grosof-2023-mama,
    author = {Grosof, I. and Harchol-Balter, M. and Scheller-Wolf, A.},
    title = {{New Stability Results for Multiserver-job Models via Product-form Saturated Systems}},
    year = {2023},
    publisher = {Association for Computing Machinery},
    volume = {51},
    journal = {SIGMETRICS Performance Evaluation Review},
    pages = {6--8}
}

@article{grosof-2024-marcreset,
    title = {{The RESET and MARC techniques, with application to multiserver-job analysis}},
    journal = {Performance Evaluation},
    volume = {162},
    pages = {102378},
    year = {2023},
    author = {Grosof, I. and Hong, Y. and Harchol-Balter, M. and Scheller-Wolf, A.}
}

@article{grosof-2022-pomacs,
    author = {Grosof, I. and Scully, Z. and Harchol-Balter, M. and Scheller-Wolf, A.},
    title = {{Optimal Scheduling in the Multiserver-job Model under Heavy Traffic}},
    year = {2022},
    volume = {6},
    journal = {Proceedings of the ACM on Measurement and Analysis of Computing Systems },
    pages = {1--32}
}

@article{MJQM,
    author = {Harchol-Balter, M.},
    title = {{The Multiserver Job Queueing Model}},
    year = {2022},
    volume = {100},
    journal = {Queueing Systems: Theory and Applications},
    pages = {201--203},
}

@inproceedings{slurm,
    author={Jette, M. A. and Wickberg, T.},
    title={{Architecture of the Slurm Workload Manager}},
    booktitle={Job Scheduling Strategies for Parallel Processing},
    year={2023},
    publisher={Springer Nature Switzerland},
    address = {Cham, Switzerland},
    pages = {3--23}
}

@article{loynes, 
    title={The stability of a queue with non-independent inter-arrival and service times}, 
    volume={58},  
    journal={Mathematical Proceedings of the Cambridge Philosophical Society}, 
    author={Loynes, R. M.}, 
    year={1962}, 
    pages={497--520}
}

@inproceedings{maguluri-infocom-2012,
    author={Maguluri, S. T. and Srikant, R. and Ying, L.},
    booktitle={IEEE INFOCOM}, 
    title={{Stochastic models of load balancing and scheduling in cloud computing clusters}},
    year={2012},
    publisher = {IEEE},
    address = {New York, USA},
    pages={702--710}
}

@inproceedings{morozov,
    author = {Morozov, E. and Rumyantsev, A. S.},
    title = {{Stability Analysis of a {MAP/M/s} Cluster Model by Matrix-Analytic Method}},
    year = {2016},
    booktitle = {EPEW}, 
    publisher = {Springer},
    address = {New York, USA},
    pages = {63--76}
}

@article{feitelson-2001,
    author={Mu'alem, A.W. and Feitelson, D.G.},
    journal={IEEE Transactions on Parallel and Distributed Systems}, 
    publisher = {IEEE Press},
    title={{Utilization, Predictability, Workloads, and User Runtime Estimates in Scheduling the IBM SP2 with Backfilling}}, 
    year={2001},
    volume={12},
    pages={529--543},
}

@article{ourPEVA,
    author = {Olliaro, D. and {Ajmone Marsan}, M. and Balsamo, S. and Marin, A.},
    title = {{The saturated Multiserver Job Queuing Model with two classes of jobs: Exact and approximate results}},
    journal = {Performance Evaluation},
    volume = {162},
    pages = {102370},
    year = {2023},
}

@article{ourTPDS,
    author={Olliaro, D. and Anggraito, A. and {Ajmone Marsan}, M. and Balsamo, S. and Marin, A.},
    journal={IEEE Transactions on Parallel and Distributed Systems}, 
    title={The Impact of Service Demand Variability on Data Center Performance}, 
    year={2025},
    volume={36},
    pages={120--132}
}

@article{rumyantsev-2017,
    author = {Rumyantsev, A. and Morozov, E.},
    year  = {2017},
    title = {{Stability criterion of a multiserver model with simultaneous service}},
    journal = {Annals of Operational Research},
    volume = {252},
    pages = {29--39}
}

@techreport{berkeley,
    title = {{2024 United States Data Center Energy Usage Report}},
    author = {Shehabi, A. and Smith, S. J. and Hubbard, A. and Newkirk, A. and Lei, N. and Siddik, M. A. and Holecek, B. and Koomey, J. and Masanet, E. and Sartor, D.},
    institution = {Berkeley Lab - Energy Analysis \& Environmental Impacts Division},
    type = {Report},
    year = {2024},
    url = {https://doi.org/10.71468/P1WC7Q}
}

@article{bichler,
    author = {Speitkamp, B. and Bichler, M.},
    journal = {IEEE Transactions on Services Computing},
    pages = {266--278},
    title = {{A Mathematical Programming Approach for Server Consolidation Problems in Virtualized Data Centers}},
    volume = {3},
    year = {2010}
}

@inproceedings{borg,
    author = {M. Tirmazi and A. Barker and N. Deng and M. E. Haque and Z. G. Qin and S. Hand and M. Harchol{-}Balter and J. Wilkes},
    title = {{Borg: the next generation}},
    year = {2020},
    booktitle = {ACM EuroSys},
    publisher = {Association for Computing Machinery},
    address = {New York, USA},
    pages = {1--14},
}

@inproceedings{yarn,
    author = {Vavilapalli, V. K. and Murthy, A. C. and Douglas, C. and Agarwal, S. and Konar, M. and Evans, R. and Graves, T. and Lowe, J. and Shah, H. and Seth, S. and Saha, B. and Curino, C. and O'Malley, O. and Radia, S. and Reed, B. and Baldeschwieler, E.},
    title = {{Apache Hadoop {YARN}: yet another resource negotiator}},
    year = {2013},
    booktitle = {ACM SOCC},
    address = {New York, USA},
    pages = {1--16},
    publisher = {ACM}
}

@inproceedings{baiocchi,
    author={Yildiz, M. and Baiocchi, A.},
    booktitle={2024 IEEE 25th International Conference on High Performance Switching and Routing (HPSR)}, 
    publisher = {IEEE},
    title = {{Data-Driven Workload Generation Based on Google Data Center Measurements}}, 
    year = {2024},
    address = {New York, USA},
    pages = {143--148}
}

\appendix

\section{On the application of the sandwich algorithm for perfect sampling}\label{app:sandwich}

As explained in Section~\ref{subsec:perfalg}, the algorithm for the production of samples of the stationary workload vector only provides SPS. Remember that a SPS is a lower bound $L_n$ to the perfect sample, indexed by a parameter $n$ linked to simulation time, and such that the sequence $L_n$ increases and couples in finite time a.s. to the perfect sample. A natural question in this context is whether there exists a stopping rule that would allow one to detect the smallest $n$ such that there is certitude that the current lower bound is a perfect sample, so that one can stop simulating. This problem is solved in~\cite{blanchet0} for the coupling from the past of the stable GI/GI/1 queue and in~\cite{blanchet} for that of the GI/GI/c queue. In the GI/GI/c case, the technique leverages the upper-bounding of the workload vector by a computable stationary majorant. When the backward simulation starting from this stationary majorant and from the empty system meet at time 0, the simulation can be stopped and shown to deliver a perfect sample. The stationary upper bound is a random assignment queuing system where the $c$ servers are loaded at random rather than adaptively (namely, less loaded first). A key point here is that the stability condition
of the GI/GI/c queue is the same as that of the associated random assignment queuing system. A natural questions is whether the very same technique could be extended to the multiserver job queuing dynamics. Unfortunately, the answer is negative in general. The reason is that the stability condition of the random assignment multiserver job queuing system is in general different (and more restrictive) than that of the genuine multiserver job queuing system as shown below. An interesting (and open) question is whether one can devise another computable stationary majorant for the multiserver job queuing workload vector that would then lead to another stopping rule. This important theoretical
question is not addressed in the present paper. In practice, for all systems simulated so far and with service times having exponential moments, the exponential backoff procedure proposed in the algorithm provides answers that very closely match the results obtained either by analytical methods or by direct discrete event simulation. In other terms, at least for this exponential moment service time case, the use of SPS and exponential backoff allows one to safely predict the performance of multiserver job queues of practical interest.

Now, we give a counterexample aimed at showing that the approach adopted in~\cite{blanchet} cannot be used here.
Consider a system where jobs are i.i.d. and require either one or all servers (jobs are called small and big in the two cases).
Let $p$ denote the probability that a job is big. Assume that all service times are deterministic of size 1. To show that the stability for the RA case is strictly more severe than that for the Greedy case, we evaluate the   pile growth rate in both cases. For this, we use our identification of the stability condition in terms of the growth rate of piles. In both cases, we can decompose the   pile into cycles.
A cycle starts with a big job and ends with the next big job. Given that the next big job arrives as the $k+1$-st job, the sub-pile of 
small jobs has size $G_k= \lfloor k/c \rfloor$ in the greedy case, and $R_k =max(m(1,k),\ldots,m(c,k))$ in the RA case, with $(m(1,k),\ldots,m(c,k))$ a multinomial random vector of parameters $k$, $1/c$. Clearly $R_k \ge G_k$ and, with a positive probability, $R_k>G_k$. Hence $E(R_k)>G_k$. Let $p_k= (1-p)^k p$, $k\ge 0$ be the (geometric) probability that there are $k$ small jobs in the cycle.
Let $MG= \sum_k p_k G_k$ and $MR= \sum_k p_k E(R_k)$. 
Then, the ratio of the growth rates of the overall piles in the RA and Greedy cases is $(1+MR)/(1+MG) >1$, which concludes the proof.

Numerical examples of the differences in the stability condition boundaries obtained with FCFS and RA in some of the cases studied in Section~\ref{sec:numerical} are presented in Table~\ref{tab:fcfsvsra}.

\begin{table}[t]
\caption{Comparison of FCFS and RA stability region boundaries, under different scenarios. The 5-classes scenario is further described in Subsection~\ref{subsec:synth}, whereas cellB setting is defined in detail in Subsection~\ref{subsec:cellB}. \label{tab:fcfsvsra}}
{\footnotesize
\vspace{-0.1in}
\label{tab:sr-comp}
\centering
\begin{tabular}{|l|c|c|c|c|c|c|}
\hline
\textbf{Scenario} & \textbf{FCFS} & \textbf{RA}\\
\hline
\hline
\textbf{$\pmb{5-}$Classes, Exponential} & 1.80 & 0.95\\
\textbf{$\pmb{5-}$Classes, Bounded Pareto} & 1.72 & 0.92\\
\textbf{Google Borg CellB, Exponential} & 2.38 & 0.85\\
\hline
\hline
\end{tabular}}
\end{table}

\end{document}